\documentclass[11pt,reqno]{amsart}
\allowdisplaybreaks[4]
\usepackage{graphicx}  
\usepackage{epsfig}
\usepackage{amssymb}
\usepackage{amsmath}
\usepackage{cite}
\usepackage{tikz}
\usepackage{xy}
\usepackage{booktabs}  
\usepackage{array}
\usepackage{arydshln}
 
\newtheorem{theorem}{Theorem}

\newtheorem{proposition}{Proposition}
\newtheorem{corollary}{Corollary}

\newtheorem{lemma}{Lemma}
\newtheorem{remark}{Remark}
\newtheorem{example}{Example}
\newcommand{\be}{\begin{equation}}
\newcommand{\ee}{\end{equation}}
\newcommand{\bea}{\begin{eqnarray}}
\newcommand{\eea}{\end{eqnarray}}
\newcommand{\ba}{\begin{array}}
\newcommand{\ea}{\end{array}}
\newcommand{\bean}{\begin{eqnarray*}}
\newcommand{\eean}{\end{eqnarray*}}

\newcommand{\La}{\Lambda}

\newcommand{\pa}{\partial}

\begin{document}

\title{The  transformations of the mToda hierarchy in tau functions}
\author{Wenchuang Guan$^1$, Shen Wang$^2$, Bailin Zhang$^1$,   Jipeng Cheng$^{3*}$}
\dedicatory { $^1$ School of Mathematics, China University of
Mining and Technology, \\
Xuzhou, Jiangsu 221116, P.\ R.\ China\\
 $^2$School of Mathematics and Statistics, Xuzhou University of Technology, \\ Xuzhou, Jiangsu 221018, P.\ R.\ China\\
 $^3$ School of Mathematical Sciences, Huaqiao University, \\
 Quanzhou, Fujian 362021, P.\ R.\ China}
\thanks{$^*$Corresponding author. Email: chengjp@hqu.edu.cn \& chengjipeng1983@163.com.}
\begin{abstract}
In this paper, we investigate the modified Toda (mToda) hierarchy, which can be regarded as the 2-component first modified Kadomtsev-Petviashvili (mKP) hierarchy. We first investigate the  connection between the Toda and mToda  tau  functions. Based on this, we construct the transformations for the mToda  tau functions  and Lax operators. Furthermore, we present the mToda squared eigenfunction symmetries  and derive the Adler-Shiota-van Moerbeke (ASvM) formula, which plays a crucial role by connecting the actions of the additional symmetries on the wave functions with the Sato--B\"{a}cklund transformations of the tau functions. Finally, by establishing the equivalence between the actions of vertex operators on the mToda tau functions and the multi-step mToda transformations, we derive the mToda addition formulas, also known as the generalized Fay identities.
\end{abstract}
\maketitle

\tableofcontents
\section{Introduction}
The Kadomtsev-Petviashvili (KP)  hierarchy is fundamentally significant in the theory of integrable systems, with broad applications spanning multiple disciplines such as partial differential equations\cite{DJKM}, the Gromov-Witten theory \cite{Witten1990}, topological recursion \cite{Alexandrov2026,Alexandrov20261,Wang20251},  Hurwitz numbers \cite{Alexandrov2023}, infinite-dimensional Lie algebras \cite{Kac2023,Jimbo,Bourgine2024}, and so on \cite{Kupershmidt1985,Kupershmidt1995,Harnad2021,Cao2026,Krichever2021,Zabrodin225}.  The multi-component generalizations are very important in  KP theory \cite{Kac2023,Lui2024,Ueno1984,Takebe2025,Savchenko2026,Savchenko20251,Savchenko20261}. Here we are interested in  the modified Toda (mToda) hierarchy, which  can be regarded as the 2-component first  modified Kadomtsev-Petviashvili  (mKP)  hierarchy \cite{Rui2024}.
 \subsection{The mToda hierarchy} The mToda hierarchy is defined by the following bilinear equation:
\begin{align}\label{HirotamToda}
&\oint_{C_{R}}\frac{dz}{2\pi {\rm i}  }z^{n-n'-1}\tau_{0, n}(\mathbf{t}-[z^{-1}]_1)
\cdot\tau_{1, n'}(\mathbf{t}'+[z^{-1}]_1)e^{\xi(\mathbf{t}^{(1)}-\mathbf{t}^{(1)'},z)}\nonumber\\
&\quad\quad\quad
+\oint_{C_{r}}\frac{dz}{2\pi{\rm i}  }z^{n-n'}\tau_{0, n+1}(\mathbf{t}-[z]_2)
\cdot\tau_{1,n'-1}(\mathbf{t}'+[z]_2)e^{\xi(\mathbf{t}^{(2)}-\mathbf{t}^{(2)'},z^{-1})}=\tau_{1,n}
(\mathbf{t})\tau_{0,n'}(\mathbf{t'}),
\end{align}
where     $\mathbf{t}=(\mathbf{t}^{({1})},\mathbf{t}^{(2)})$, $\mathbf{t}^{(a)}=(t^{(a)}_1,t_2^{(a)}, \ _{\cdots})$,  $[z^{-1}]=(z^{-1},z^{-2}/2, \ _{\cdots})$, $\mathbf{t}-[z^{-1}]_a=
(\mathbf{t}^{({1})},\mathbf{t}^{(2)})
|_{{\mathbf{t}^{(a)}\rightarrow}\mathbf{t}^{(a)}-[z^{-1}]}$, and $\xi(\mathbf{t}^{(a)},z)=\sum_{j\geq1}t^{(a)}_jz^j$ for $a=1,2$.
$C_R$
  denotes the anticlockwise circular contour
$|z|=R$ for sufficiently large $R$, and $C_r$ denotes the anticlockwise circular contour
$|z|=r$ for sufficiently small $r$.
   Here $(\tau_{0,n}(\mathbf{t}),\tau_{1,n}
(\mathbf{t}))$ is called the mToda tau pair.

If we introduce the mToda wave functions $\Psi_a(n,\mathbf{t},z)$ and the mToda adjoint wave functions $\Psi_a^{*}(n,\mathbf{t},z)$ $(a=1,2)$,
\begin{align}
		&\Psi_1(n,\mathbf{t}, z)=\frac{ \tau_{0,n}(\mathbf{t}-[z^{-1}]_1) }{ \tau_{1,n}(\mathbf{t})  }z^{n}e^{\xi(\mathbf{t}^{(1)}, z)},  \quad
		\Psi_2(n,\mathbf{t}, z)= \frac{ \tau_{0,n+1}(\mathbf{t}-[z]_2) }{ \tau_{1,n}(\mathbf{t})  }z^{n}e^{\xi(\mathbf{t}^{(2)},z^{-1})}, \label{mToadwave1}\\
		&\Psi_1^{*}(n,\mathbf{t}, z)=\frac{ \tau_{1,n}(\mathbf{t}+[z^{-1}]_1) }{ \tau_{0,n}(\mathbf{t})  }z^{-n}e^{-\xi(\mathbf{t}^{(1)}, z)}, \quad
		\Psi_2^{*}(n,\mathbf{t}, z)=\frac{ \tau_{1,n-1}(\mathbf{t}+[z]_2) }{ \tau_{0,n}(\mathbf{t})  }z^{-n+1}e^{-\xi(\mathbf{t}^{(2)},z^{-1})},\label{mToadwave2}
	\end{align}
 then the mToda bilinear equation \eqref{HirotamToda} is equivalent to:
\begin{align}
\oint_{C_R}\frac{dz}{2\pi {\rm i}z  }\Psi_1(n,\mathbf{t},z)\Psi_1^*(n',\mathbf{t}',z)+\oint_{C_r}\frac{dz}{2\pi {\rm i}z }\Psi_2(n,\mathbf{t},z)\Psi_2^*(n',\mathbf{t}',z)=1.\label{mTodaWavbi}
\end{align}

If we  define the mToda wave operators
\begin{eqnarray}
&&S_1(n,\mathbf{t},\La)=\sum_{i\geq 0}\frac{p_i(-\widetilde{\partial}^{(1)})\tau_{0,n}(\mathbf{t})}
{\tau_{1,n}(\mathbf{t})}\Lambda^{-i},
\ \ S_2(n,\mathbf{t},\La)=\sum_{i\geq 0}\frac{p_i(-\widetilde{\partial}^{(2)})\tau_{0,n+1}(\mathbf{t})}
{\tau_{1,n}(\mathbf{t})}\Lambda^{i}, \label{mTodaderess1}
\end{eqnarray}
where $\widetilde{\partial}^{(a)}=(\partial_{t_1^{(a)}},
\partial_{t_2^{(a)}}/2, \ _{\cdots}\ ,\partial_{t_k^{(a)}}/k, \ ...)$, $\Lambda$ is the shift operator satisfying $\Lambda(f(n))=f(n+1)$, and $p_i(\mathbf{t})$ is the Schur polynomial defined by $e^{\xi(\mathbf{t},z)}=\sum_{i=0}^{+\infty}p_i(\mathbf{t})z^i$  and $p_i(\mathbf{t})=0$ if  $i<0$, then the wave operators $S_1$ and $S_2$ satisfy the following evolution equations
\begin{align}
&\pa_{t_k^{(1)}}{S}_1=-( S_1\La^kS_1^{-1})_{\Delta,\leq0}{S}_1,\quad \pa_{t_k^{(1)}}{S}_2=( S_1\La^kS_1^{-1})_{\Delta,\geq1}{S}_2,\label{mTodaSx1} \\
&\pa_{t_k^{(2)}}{S}_1=( S_2\La^{-k}S_2^{-1})_{\Delta^{*},\geq1}{S}_1,\quad
\pa_{t_k^{(2)}}{S}_2=-( S_2\La^{-k}S_2^{-1})_{\Delta^{*},\leq0}{S}_2,\label{mTodaSx2}
\end{align}
where  $\Delta=\Lambda-1$, $\Delta^*=\Lambda^{-1}-1$, and $\sum_{i}(a_iP^i)_{P, \geq k}=\sum_{i\geq k}a_iP^i$ with $P=\Delta$ or $\Delta^*$.
Further, the mToda wave functions $\Psi_a(n,\mathbf{t},z)$ and the mToda adjoint wave functions $\Psi_a^*(n,\mathbf{t},z)$   can be expressed   in the following manner
\begin{align*}
&\Psi_1(n,\mathbf{t},z)=S_1(n,\mathbf{t},\La)
e^{\xi({\mathbf{t}}^{(1)},\La)}(z^n), \quad\Psi_2(n,\mathbf{t},z)=S_2(n,\mathbf{t},\La)
e^{\xi({\mathbf{t}}^{(2)},\La^{-1})}(z^n), \\
&\Psi_1^*(n,\mathbf{t},z)=-\iota_{\La}\Delta^{-1}(S_1^{-1})^*e^{-\xi({
\mathbf{t}}^{(1)},\La^{-1})}(z^{-n}),\quad
\Psi_2^*(n,\mathbf{t},z)=\iota_{\La^{-1}}\Delta^{-1}(S_2^{-1})^*e^{-\xi({
\mathbf{t}}^{(2)},\La)}(z^{-n}),
\end{align*}
where $\iota_{\La}\Delta^{-1}=-\sum_{i\geq 0}\Lambda^i$ and $\iota_{\La^{-1}}\Delta^{-1}=\sum_{i\geq 1}\Lambda^{-i}$.
If we define  the mToda Lax operators $L_1=S_1\cdot \La\cdot S_1^{-1}$, $L_2=S_2\cdot\La^{-1}\cdot S_2^{-1}$, then the mToda Lax operators $L_1$ and $L_2$ have the following forms
\begin{align}
L_1=e^{{\beta_1}(n,\mathbf{t})}\Lambda+\sum_{i=0}^{\infty}u_i(n,\mathbf{t})
\La^{-i},\quad
L_2
=e^{{\beta_2}(n,\mathbf{t})}
\Lambda^{-1}+\sum_{i=0}^{\infty}\bar{u}_i(n,\mathbf{t})\La^{i}.\label{Laxoperadj}
\end{align}
The  mToda Lax operators $L_1$ and $L_2$ satisfy the following   Lax equations
\begin{eqnarray}\label{2mkplax}
\pa_{t_k^{(a)}}L_b=[B_k^{(a)},L_b], \quad B_k^{(1)}=( L_1^k)_{\Delta,\geq1},\quad
		B_k^{(2)}=( L_2^{k})_{\Delta^{*},\geq1} \quad a,b=1,2.
\end{eqnarray}
 The mToda hierarchy contains the mToda lattice equation, also known as the (2+1)--dimensional derivative Toda equation \cite{Hirota2004,Cao2008}:
	\begin{eqnarray*}
		\pa_{{t}^{(1)}_1}\pa_{{t}^{(2)}_1}\varphi(n)+	\big(e^{\varphi(n+1)-\varphi(n)}-e^{\varphi(n)
-\varphi(n-1)}\big)\pa_{{t}^{(1)}_1}\varphi(n)=0.
	\end{eqnarray*}
If we add the constraint below on the mToda Lax operators
\begin{align*}
L^*_2(\Lambda-\Lambda^{-1})=(\Lambda-\Lambda^{-1})L_1,
\end{align*}
then we can get the large BKP hierarchy, also called the B-Toda hierarchy \cite{Guan2025,Krichever2023},  which only depends on $\hat{\mathbf{t}}_k=({t}^{(1)}_k-{t}^{(2)}_k )/2$.
\subsection{The Toda hierarchy}
In this subsection,   we    review the fundamental aspects of the Toda hierarchy \cite{Takebe2025,Savchenko2026,Ueno1984}.
By the Miura transformations, the mToda hierarchy can be related to the Toda hierarchy (for more details  refer to \cite{Rui2024}). The Toda Lax equations are defined by
\begin{align}
&\pa_{t_k^{(a)}}\mathcal{L}_b=[ \mathcal{B}_k^{(a)} ,\mathcal{L}_b], \quad \mathcal{B}_k^{(1)}=(\mathcal{L}^k_1)_{\geq 0},\quad \mathcal{B}_k^{(2)}=(\mathcal{L}_2^k)_{< 0},\quad a,b=1,2,\label{TodaLax}
\end{align}
where  $(\sum_{i}a_i\Lambda^i)_{\geq k}=\sum_{i\geq k}a_i\Lambda^i$, $(\sum_{i}a_i\Lambda^i)_{< k}=\sum_{i< k}a_i\Lambda^i$, and the Toda Lax operators $ \mathcal{L}_1$ and $ \mathcal{L}_2$ have the following forms
\begin{align*}
 \mathcal{L}_1=\Lambda +\sum_{i=0}^{\infty}v_i(n,\mathbf{t})\Lambda^{-i},\quad
{\mathcal{L}_2}= e^{\alpha(n,\mathbf{t})}\Lambda^{-1} +\sum_{i=0}^{\infty}  \overline{v}_i(n,\mathbf{t})\Lambda^{i}.
 \end{align*}
For the Toda hierarchy, there exists a  tau function $\tau^{{\rm Toda}}_n(\mathbf{t})$   \cite{takasaki2018jpa},
such that the whole Toda hierarchy   can be expressed by the bilinear equation below:
\begin{align}
&\oint_{C_{R}}\frac{dz}{2\pi {\rm i} }z^{n-n'}\tau^{{\rm Toda}}_n(\mathbf{t}-[z^{-1}]_1)\tau^{{\rm Toda}}_{n'}(\mathbf{t}'+[z^{-1}]_1)
e^{\xi({\mathbf{t}}^{(1)}-{\mathbf{t}}^{(1)\prime},z)}\nonumber\\
\qquad\qquad&=
\oint_{C_{r}}\frac{dz}{2\pi {\rm i} }z^{n-n'}\tau^{{\rm Toda}}_{n+1}(\mathbf{t}-[z]_2)\tau^{{\rm Toda}}_{n'-1}(\mathbf{t}'+[z]_2)
e^{\xi({\mathbf{t}}^{(2)}-{\mathbf{t}}^{(2)\prime},z^{-1})}.\label{todabilitau}
\end{align}
On this basis, we can define the Toda wave functions $\psi_a(n,\mathbf{t},z)$ and Toda adjoint wave functions $\psi^*_a(n,\mathbf{t},z)$ for $(a=1,2)$   by the Toda tau function $\tau^{{\rm Toda}}_n(\mathbf{t})$ \cite{Ueno1984}:
\begin{align}
&\psi_1(n,\mathbf{t},z)=\frac{\tau^{{\rm Toda}}_n
(\mathbf{t}-[z^{-1}]_1)}{\tau^{{\rm Toda}}_n(\mathbf{t})}e^{\xi({\mathbf{t}}^{(1)},z)}z^n,\quad
{\psi_2}(n,\mathbf{t},z)=\frac{\tau^{{\rm Toda}}_{n+1}(\mathbf{t}-[z]_2)}{\tau^{{\rm Toda}}_n(\mathbf{t})}
e^{\xi({\mathbf{t}}^{(2)},z^{-1})}z^{n},\label{Todawave1}
\\
&{\psi_1}^*(n,\mathbf{t},z)=\frac{\tau^{{\rm Toda}}_{n+1}(\mathbf{t}+[z^{-1}]_1)}{\tau^{{\rm Toda}}_{n+1}(\mathbf{t})}e^{-\xi({\mathbf{t}}^{(1)},z)}z^{-n},\quad
{\psi_2}^*(n,\mathbf{t},z)=\frac{\tau^{{\rm Toda}}_{n}(\mathbf{t}+[z]_2)}{\tau^{{\rm Toda}}_{n+1}(\mathbf{t})}e^{-\xi({\mathbf{t}}^{(2)},z^{-1})}z^{-n},\label{Todawave2}
\end{align}
which satisfy
\begin{align}
&\mathcal{L}_1(\psi_1)=z\psi_1,\quad\mathcal{L}_2(\psi_2)=z^{-1}\psi_2,\quad
\mathcal{L}^*_1(\psi^*_1)=z\psi^*_1,
\quad\mathcal{L}^*_2(\psi^*_2)=z^{-1}\psi^*_2,\label{TLtau}\\
&\pa_{t^{(a)}_k}\big(\psi_b\big)
=(\mathcal{B}_k^{(a)})\big(\psi_b\big),\quad
\pa_{t^{(a)}_k}\big(\psi_b^{*}\big)
=-(\mathcal{B}_k^{(a)})^{*}\big(\psi_b^*\big),\quad  a,b=1,2.\label{todaawavetau}
\end{align}
More generally, the Toda wave functions $\psi_a(n,\mathbf{t},z)$ and Toda adjoint wave functions $\psi^*_a(n,\mathbf{t},z)$ are the special case of the Toda eigenfunction $q_n$ and  the Toda adjoint eigenfunction $r_n$,   defined by
\begin{align}
&\pa_{t^{(a)}_k}q_n=\mathcal{B}_k^{(a)}(q_n),\quad
\pa_{t^{(a)}_k}r_n=-(\mathcal{B}_k^{(a)})^*(r_n),\quad a=1,2
,\label{Todaeig2}
\end{align}
 satisfying the following spectral representation  \cite{Lui2024}
 \begin{align}
		q_n( \mathbf{t})=&\oint_{C_R}\frac{{\rm d}z}{2\pi {\rm i}}
		z^{-2}q_{n'}(\mathbf{t}'+[z^{-1}]_1)
\psi_1^{*}(n'-1,\mathbf{t}',z)
\psi_1(n,\mathbf{t},z)\nonumber\\
		&+\oint_{C_r}\frac{{\rm d}z}{2\pi {\rm i}}z^{-1}q_{n'-1}(\mathbf{t}'+[z]_2)\psi_2^{*}(n'-1,\mathbf{t}',z)
\psi_2(n,\mathbf{t},z),\label{qToda}\\
r_n(\mathbf{t})=&\oint_{C_R}\frac{{\rm d}z}{2\pi {\rm i}}z^{-2}r_{n'-1}( \mathbf{t}'-[z^{-1}]_1)\psi_1(n',\mathbf{t}',z)
\psi_1^{*}(n,\mathbf{t},z)\nonumber\\
&+\oint_{C_r}\frac{{\rm d}z}{2\pi {\rm i}}z^{-1}r_{n'}(\mathbf{t}'-[z]_2)\psi_2(n',\mathbf{t}',z)
\psi_2^{*}(n,\mathbf{t},z).\label{rToda}
		\end{align}

\subsection{Major results} In this paper, there are the following main results. First, the relations between the mToda tau functions and the Toda tau function are given by the following theorem.
\begin{theorem}\label{Todaeigadjo}
 (I)
 Given the mToda tau pair  $(\tau_{0,n} ,\tau_{1,n} )$, $\tau_{0,n} $ and $\tau_{1,n} $  are two Toda tau functions, satisfying the Toda bilinear equation \eqref{todabilitau}.
If we define
 \begin{align*}
q_n =\tau_{1,n} /\tau_{0,n} ,\quad
r_n =\tau_{0,n+1} /\tau_{1,n+1} ,
\end{align*}
then $q_n $  is the  Toda eigenfunction  associated with $\tau_{0,n} $, while $r_n $
 acts as the Toda adjoint eigenfunction relative to  $\tau_{1,n} $.

(II) Conversely,
  given  the Toda eigenfunction $q_n $ and the Toda adjoint eigenfunction $r_n $ with respect to the  Toda tau function $\tau^{{\rm Toda}}_n $,
 \begin{align*}
\tau_n^{[1]} :=q_n \tau^{{\rm Toda}}_n ,\quad \tau_n^{[-1]} :=r_{n-1} \tau^{{\rm Toda}}_{n} ,
 \end{align*}
can be regarded as new Toda tau functions.
Furthermore,  the tau pair \begin{align*}
 (\tau_{0,n} ,\tau_{1,n} )&:=(\tau^{{\rm Toda}}_{n} , \tau_n^{[1]} ) \quad {\rm or} \quad ( \tau_n^{[-1]} ,\tau^{{\rm Toda}}_{n} ),
  \end{align*}satisfies the mToda bilinear equation \eqref{HirotamToda}.

\end{theorem}
\begin{remark}
  Although the first part of the theorem is stated in \cite{Rui2024}, the second part is  still left unfinished. Here, we establish a complete proof of the full theorem. We say that $q_n (\mathbf{t})$ and $r_n (\mathbf{t})$ are the Toda eigenfunction and adjoint eigenfunction associated with $\tau_n^{\mathrm{Toda}}(\mathbf{t}) $, meaning  that they satisfy \eqref{Todaeig2}, and
the Toda Lax operators  $\mathcal{L}_1$ and $\mathcal{L}_2$ are expressed in terms of the tau function
 $\tau_n^{\mathrm{Toda}}(\mathbf{t})$ by
\begin{align}
\mathcal{L}_1&=\sum_{l=0}^{+\infty}\left(\sum_{i,j=0,i+j=l}^{+\infty}
\frac{p_i(-\widetilde{\partial}^{(1)})\tau^{{\rm Toda}}_n(\mathbf{t})p_j(\widetilde{\partial}^{(1)})\tau^{{\rm Toda}}_{n+2-l}(\mathbf{t})}{\tau^{{\rm Toda}}_n(\mathbf{t})\tau^{{\rm Toda}}_{n+2-l}(\mathbf{t})}\right)\Lambda^{1-l},\label{TodaLax1}\\
\mathcal{L}_2&=\sum_{l=0}^{+\infty}\left(\sum_{i,j=0,i+j=l}^{+\infty}
\frac{p_i(-\widetilde{\partial}^{(2)})\tau^{{\rm Toda}}_{n+1}(\mathbf{t})p_j(\widetilde{\partial}^{(2)})\tau^{{\rm Toda}}_{n+l-1}(\mathbf{t})}{\tau^{{\rm Toda}}_n(\mathbf{t})\tau^{{\rm Toda}}_{n+ l}(\mathbf{t})}\right)\Lambda^{l-1}.\label{TodaLax2}
\end{align}

\end{remark}
Let $\widetilde{q}_{n} $ and $\widetilde{r}_{n} $ be the mToda  eigenfunction and the mToda adjoint eigenfunction, respectively, defined by
	\begin{align}
		 &\pa_{t_k^{(1)}}\widetilde{q}_{n}
=B_k^{(1)}(\widetilde{q}_{n} ),\quad \pa_{t_k^{(2)}}\widetilde{q}_{n}
=B_k^{(2)}(\widetilde{q}_{n} ),\label{eigenfunction}\\ &\pa_{t_k^{(1)}}\widetilde{r}_{n}
=-(\iota_{\La^{-1}}\Delta^{-1}B_k^{(1)*}\Delta)(\widetilde{r}_{n} ), \quad \pa_{t_k^{(2)}}\widetilde{r}_{n}
=-(\iota_{\La}\Delta^{-1}B_k^{(2)*}\Delta)(\widetilde{r}_{n} ),\label{adjointeigenfunction}
	\end{align}
 then we have the following theorem for the mToda  tau pair.
\begin{theorem}\label{mTodaeigentran}
(I) Given the mToda tau pair $(\tau_{0, n} ,\tau_{1,n} )$, the mToda  eigenfunction $\widetilde{q}_{n} $ and the mToda adjoint eigenfunction $\widetilde{r}_{n} $, if we define  $\tau^{[1]}_{0,n} $ and $\tau^{[1]}_{1,n} $ by the   following ways:

\begin{itemize}
  \item $\textbf{Case 1:}$
 \begin{align*}\tau_{0,n}^{[1]} =\widetilde{q}_n \tau_{1,n}, \quad \tau_{1,n}^{[1]} =\dfrac{ (\widetilde{q}_{n}-\widetilde{q}_{n+1} )\tau_{1,n} \tau_{1,n+1} } {\tau_{0,n+1} },\end{align*}
  \item $\textbf{Case 2:}$
 \begin{align*}\tau_{0,n}^{[1]} =\dfrac{ (\widetilde{r}_{n}-\widetilde{r}_{n-1} )\tau_{0,n}  \tau_{0,n-1} }{\tau_{1,n-1} }, \quad \tau_{1,n}^{[1]} =\widetilde{r}_n \tau_{0,n},\end{align*}
\end{itemize}
 then   $(\tau_{0, n}^{[1]} ,\tau_{1,n}^{[1]} )$ is still a new  mToda tau pair. Specifically, it satisfies the mToda bilinear equation \eqref{HirotamToda}.

 (II) If we denote   $ L_1$ and $L_2$  as the mToda  Lax operators with respect to the tau pair $(\tau_{0,n} ,\tau_{1,n} )$ by \eqref{mTodaderess1}, and  denote $L_a^{[1]}=T_b\cdot L_a\cdot T_b^{-1}$ for $(a,b=1,2)$
with the operators \begin{align*}
T_1=\frac{\widetilde{q}_n\widetilde{q}_{n+1}}{\widetilde{q}_n-\widetilde{q}_{n+1}}
\cdot\Delta\cdot\widetilde{q}_n^{-1},\quad T_2=\widetilde{r}_n^{-1}\cdot\Delta^{-1}
\cdot (\widetilde{r}_{n+1}-\widetilde{r}_{n}),
\end{align*}then  $L_1^{[1]}$ and $L_2^{[1]}$ are new   mToda   Lax operators.
 Here $\Delta^{-1}$ in $T_1^{-1}$ or  $T_2$    are expanded   in the form    $\sum_{ -\infty\ll i }a_i(n)\Lambda^{i}$ for $L_1^{[1]}=T_b\cdot L_1\cdot T_b^{-1}$,   while for $L_2^{[1]}=T_b\cdot L_2\cdot T_b^{-1}$,  they are expressed in the form $\sum_{i\ll +\infty}a_i(n)\Lambda^i$.

(III) Further given another  mToda eigenfunction  $\widetilde{f}_n$ and   mToda adjoint eigenfunction $\widetilde{g}_n$, $\widetilde{f}^{[1]}_n=T_b(\widetilde{f}_n)$  is the  mToda eigenfunction and $\widetilde{g}^{[1]}_n = \big(\Delta^{-1}(T_b^*)^{-1}\Delta\big)(\widetilde{g}_n)$  is the  mToda adjoint eigenfunction corresponding to $(\tau^{[1]}_{0,n}, \tau^{[1]}_{1,n})$ in $\textbf{Case b}$ for $b=1,2$.

\end{theorem}

If   $T_2$ is applied to the mToda tau pair
 $(\tau_{0, n},\tau_{1,n})$, followed by    $T_1$,
then we have the following corollary.

\begin{corollary}\label{sepmToda}
Given the mToda tau pair $(\tau_{0, n},\tau_{1,n})$, the mToda eigenfunction $\widetilde{q}_n$, the mToda adjoint eigenfunction $\widetilde{r}_n$,
 \begin{align*}
  \Big(\tau^{\{1\}}_{0,n},\tau^{\{1\}}_{1,n}\Big):=
 \left(\Omega\Big(\widetilde{q}_n,
 \Delta(\widetilde{r}_n)\Big)\tau_{0,n},\Omega\Big( \Delta(\widetilde{q}_n),
\widetilde{r}_{n+1}\Big)\tau_{1,n}\right),
 \end{align*}
 is also the mToda tau pair, satisfying the mToda bilinear equation \eqref{HirotamToda}. Here  $\Omega\Big(\widetilde{q}_n,
 \Delta(\widetilde{r}_n)\Big)$ is the mToda    squared eigenfunction potential (SEP)  (for more details see Appendix).
\end{corollary}

If we set    $\widetilde{q}_n(\mathbf{t})=\Psi_a(n,\mathbf{t},\mu^{\eta_{a}})$, $\widetilde{r}_n(\mathbf{t})=\Psi_b^*(n,\mathbf{t},\lambda^{\eta_b})$, $\eta_i=(-1)^{\delta_{i,2}}$   in Corollary \ref{sepmToda}, and   introduce the following vertex operators \cite{Adler1999}
\begin{align*}
&X_1(n,\mathbf{t},\mu):=e^{\xi(\mathbf{t}^{(1)},\mu)}
e^{-\xi(\widetilde{\partial}^{(1)},\mu^{-1})}\mu^n,\quad
X_1^*(n,\mathbf{t},\lambda):=-\lambda^{-n}e^{-\xi(\mathbf{t}^{(1)},\lambda)}
e^{\xi(\widetilde{\partial}^{(1)},\lambda^{-1})},\\
&X_2(n,\mathbf{t},\mu):=-e^{\xi(\mathbf{t}^{(2)},\mu)}
e^{-\xi(\widetilde{\partial}^{(2)},\mu^{-1})}\mu^{-n}\Lambda,\quad
X_2^*(n,\mathbf{t},\lambda):=\Lambda^{-1}\lambda^{n}
e^{-\xi(\mathbf{t}^{(2)},\lambda)}e^{\xi(\widetilde{\partial}^{(2)},\lambda^{-1})},\\
&\mathbb{X}_{ab}(n,\mathbf{t},\lambda,\mu):=X_b^*(n,\mathbf{t},\lambda)X_a(n,\mathbf{t},\mu),\quad a,b=1,2,
\end{align*}
then we can get the following theorem.

\begin{theorem}\label{ASvM}
(I) For the mToda tau pair   $(\tau_{0,n},\tau_{1,n})$, if we define
 \begin{align*}
\tilde{\tau}_{i,n}^{ab}=\tau_{i,n}+C_{ab}\mathbb{X}_{ab}
(n,\mathbf{t},\lambda,\mu)\tau_{i,n},\quad i=0,1,\quad a,b=1,2,
 \end{align*}
 where $C_{ab}$ is some constant, then $(\tilde{\tau}_{0,n}^{ab},\tilde{\tau}_{1,n}^{ab})$ can be regarded as a new  mToda tau pair, satisfying the mToda bilinear equation \eqref{HirotamToda}.

(II) Furthermore, if we define the additional  flows
${\pa}_{{\lambda,\mu}}^{ab}$ on the mToda wave functions $\Psi_c(n,\mathbf{t},z)$:
\begin{align*}
{\pa}_{{ \lambda,\mu}}^{ab}\Psi_c(n,\mathbf{t},z)
&=\Psi_a(n,\mathbf{t},\mu^{\eta_a})
\Omega\Big(\Delta (\Psi_c(n,\mathbf{t},z)),\Psi^*_{b}(n+1,\mathbf{t},\lambda^{\eta_b})\Big),\quad a,b,c=1,2,
\end{align*}
then
\begin{align*}
\frac{{\pa}_{{ \lambda,\mu}}^{ab}(\Psi_1(n,\mathbf{t},z))}{\Psi_1(n,\mathbf{t},z)}
&=(-1)^{\delta_{a,b}}e^{-\xi(\widetilde{\partial}^{(1)},z^{-1})}
\frac{\mathbb{X}_{ab}(n,\mathbf{t},\lambda,\mu)
\tau_{0,n }(\mathbf{t})}{\tau_{0,n }(\mathbf{t})}
+\frac{\mu^{\delta_{a,1}}}{\lambda^{\delta_{b,1}}}\frac{
\mathbb{X}_{ab}(n,\mathbf{t},\lambda,\mu)
\tau_{1,n}(\mathbf{t}) }{\tau_{1,n}(\mathbf{t})},\\
\frac{{\pa}_{{ \lambda,\mu}}^{ab}(\Psi_2(n,\mathbf{t},z))}{\Psi_2(n,\mathbf{t},z)}
&=(-1)^{\delta_{a,b}}e^{-\xi(\widetilde{\partial}^{(2)},z)}
\frac{\mathbb{X}_{ab}(n+1,\mathbf{t},\lambda,\mu)
\tau_{0,n+1}(\mathbf{t})}{\tau_{0,n+1}(\mathbf{t})}
+\frac{\mu^{\delta_{a,1}}}{\lambda^{\delta_{b,1}}}\frac{
\mathbb{X}_{ab}(n,\mathbf{t},\lambda,\mu)
\tau_{1,n}(\mathbf{t}) }{\tau_{1,n}(\mathbf{t})}.
\end{align*}

(III) In particular, the actions of the additional  flows ${\pa}_{ \lambda,\mu}^{ab}$ on the mToda tau functions   are  given by
\begin{align*}
&{\pa}_{\lambda,\mu}^{ab}\tau_{0,n}=(-1)^{\delta_{a,b}}
\mathbb{X}_{ab}
(n,\mathbf{t},\lambda,\mu)\tau_{0,n},\quad
{\pa}_{ \lambda,\mu}^{ab}\tau_{1,n}=-
\frac{\mu^{\delta_{a,1}}}{\lambda^{\delta_{b,1}}}
\mathbb{X}_{ab}
(n,\mathbf{t},\lambda,\mu)\tau_{1,n}.
\end{align*}
\end{theorem}

\begin{remark}
Here, ${\pa}_{\lambda,\mu}^{ab}$ can also be regarded as the generators of the mToda additional symmetries. Although the case $a=b\in\{1,2\}$ of this theorem was stated in \cite{Yang2025}, the cases of $a \neq b$ were not considered. In particular, the second part of the theorem is also known as the Adler-Shiota-van Moerbeke (ASvM) formula, which connects the actions of the additional symmetries on the wave functions with those on the tau functions.
\end{remark}
\begin{theorem}\label{mTodahigerfayide}
Given the mToda tau pair  $(\tau_{0,n},\tau_{1,n})$, the mToda wave functions $\Psi_a$  and the mToda adjoint wave functions $\Psi^*_a$ $(a=1,2)$, we have
\begin{itemize}
  \item $s\geq k$
\begin{align*}
    &\left(\prod_{\gamma=i}^1 X_1^*(\lambda_\gamma) \prod_{\delta=k-i}^{1} X_2^*(\kappa_\delta) \prod_{\beta=s-l}^{1}X_2(\nu_\beta) \prod_{\alpha=l}^1 X_1(\mu_{\alpha})\right)(\tau_{\vartheta,n}) \nonumber \\
    &= (-1)^{s(i-1)+\varepsilon+(1-\vartheta)l + \vartheta i} \frac{\prod_{\gamma=1}^i \lambda_\gamma^{s-k+i-\gamma+\vartheta}}{\prod_{\alpha=1}^l \mu_\alpha^{\alpha-1+\vartheta}}
    \cdot \det \!\big( \mathbf{M}^{(\vartheta)} \big) \cdot
    \left( \prod_{m=0}^{s-k-1+\vartheta} \frac{\tau_{1,n+m}}{\tau_{0,n+m}} \right) \cdot \tau_{0,n}.
\end{align*}
  \item   $s < k$
\begin{align*}
    &\left(\prod_{\gamma=i}^1 X_1^*(\lambda_\gamma) \prod_{\delta=k-i}^{1} X_2^*(\kappa_\delta) \prod_{\beta=s-l}^{1}X_2(\nu_\beta) \prod_{\alpha=l}^1 X_1(\mu_{\alpha})\right)(\tau_{\vartheta,n}) \nonumber \\
    &= (-1)^{si+\varepsilon+(1-\vartheta)l -(k+1-i) \vartheta} \frac{\prod_{\gamma=1}^i \lambda_\gamma^{s-k+i-\gamma+\vartheta}}{\prod_{\alpha=1}^l \mu_\alpha^{\alpha-1+\vartheta}}
    \cdot \det \!\big( \widetilde{\mathbf{M}}^{(\vartheta)} \big)_{n-1} \cdot
    \left( \prod_{m=0}^{k-s-1-\vartheta} \frac{\tau_{0,n-m-1}}{\tau_{1,n-m-1}} \right) \cdot\tau_{0,n}.
\end{align*}
\end{itemize}
 Here the subscript $n-1$ after the determinant $\widetilde{\mathbf{M}}^{(\vartheta)}$ indicates that all elements of the internal matrix are $n\rightarrow n-1$,   $\vartheta\in \{0, 1\}$,
    $\varepsilon = \frac{(s-l)(s+l-1)}{2} + \frac{(k-i)(k-i-1)}{2} $ for $0\leq i\leq k$, $0\leq l\leq s$, $X_a(\varrho) := X_a(n,\mathbf{t},\varrho)$,  $X_a^*(\varrho) := X_a^*(n,\mathbf{t},\varrho)$, $\Psi_a(\varrho):=\Psi_a(n,\mathbf{t},\varrho)$,   $\Psi_a^*(\varrho):=\Psi_a^*(n,\mathbf{t},\varrho)$ for $ a \in \{1,2\}, \varrho \in \{\lambda,\mu,\kappa,\nu\} $ and
\begin{align*}
    \mathbf{M}^{(0)} := \begin{pmatrix}   {\mathbf{A}} \\  { {\mathbf{B}}} \end{pmatrix}_{s\times s},\quad
    \mathbf{M}^{(1)} := \begin{pmatrix}   {\mathbf{A}} & \mathbf{D} \\ \mathbf{C}  & \mathbf{E}\end{pmatrix}_{(s+1)\times(s+1)},\quad
    \widetilde{\mathbf{M}}^{ (0)} := \begin{pmatrix} \widetilde{\mathbf{A}} \\[4pt] \widetilde{ {\mathbf{B}}} \end{pmatrix}_{k\times k},\quad
    \widetilde{\mathbf{M}}^{(1)} := \begin{pmatrix} \widetilde{\mathbf{D}}  \\[4pt] \widetilde{\mathbf{A}}  \\[4pt] \widetilde{ {\mathbf{C}}} \end{pmatrix}_{k\times k},
\end{align*}
with the matrix  sub-blocks are given by:
\begin{small}
\begin{align*}
    {\mathbf{A}} &:= \begin{pmatrix}  {\Omega}\Big( \Psi_1(\mu_\alpha),\Delta(\Psi_1^*(\lambda_{i-\gamma+1})) \Big) _{\substack{1 \le \gamma \le i \\ 1 \le \alpha \le l}} &  {\Omega}\Big( \Psi_2(\nu^{-1}_\beta),\Delta(\Psi_1^*(\lambda_{i-\gamma+1})) \Big)_{\substack{1 \le \gamma \le i \\ 1 \le \beta \le s-l}} \\[8pt]
    {\Omega}\Big( \Psi_1(\mu_\alpha),\Delta(\Psi_2^*(\kappa^{-1}_{k-i-\delta+1})) \Big)_{\substack{1 \le \delta \le k-i \\ 1 \le \alpha \le l}} & {\Omega}\Big( \Psi_2(\nu^{-1}_\beta),\Delta(\Psi_2^*(\kappa^{-1}_{k-i-\delta+1})) \Big)_{\substack{1 \le \delta \le k-i \\ 1 \le \beta \le s-l}} \end{pmatrix}_{k\times s}, \\[8pt]
    {\mathbf{B}} &:= \begin{pmatrix} \left( \Delta^{m}\Big(\Psi_1(\mu_\alpha)\Big) \right)_{\substack{0 \le m \le s-k-1 \\ 1 \le \alpha \le  l}} & \left( \Delta^{m}\Big(\Psi_2(\nu^{-1}_\beta)\Big) \right)_{\substack{0 \le m \le s-k-1 \\ 1 \le \beta \le  s-l}} \end{pmatrix}_{(s-k)\times s}, \\[8pt]
     {\mathbf{C}} &:= \begin{pmatrix} \left( \Delta^{m}\Big(\Psi_1(\mu_\alpha)\Big) \right)_{\substack{0 \le m \le s-k \\ 1 \le \alpha \le  l}} & \left( \Delta^{m}\Big(\Psi_2(\nu^{-1}_\beta)\Big) \right)_{\substack{0 \le m \le s-k \\ 1 \le \beta \le  s-l}} \end{pmatrix}_{(s-k+1)\times s},\\[8pt]
    \mathbf{D}  &:= \begin{pmatrix} \big(\Psi_1^*(\lambda_\gamma)\big)_{\substack{1 \le \gamma \le i \\ 1}} \\[8pt] \big(\Psi_2^*(\kappa^{-1}_\delta)\big)_{\substack{1 \le \delta \le k-i \\ 1}} \end{pmatrix}_{k\times 1},  \quad
    \mathbf{E} :=\begin{pmatrix} \big(\delta_{m,0}\big)_{\substack{0 \le m \le s-k \\ 1}}\end{pmatrix}_{(s-k+1)\times 1},\\
    \widetilde{\mathbf{A}}  &:= \begin{pmatrix}
    \Big( \Lambda\Omega\big( \Psi_2(\nu^{-1}_{s-l-\beta+1}), \Delta(\Psi_2^*(\kappa^{-1}_\delta)) \big) \Big)_{\substack{1 \le \beta \le s-l \\ 1 \le \delta \le k-i}} &
    \Big( \Lambda\Omega\big(\Psi_2(\nu^{-1}_{s-l-\beta+1}), \Delta(\Psi_1^*(\lambda_\gamma)) \big) \Big)_{\substack{1 \le \beta \le s-l \\ 1 \le \gamma \le i}} \\[12pt]
    \Big( \Lambda\Omega\big( \Psi_1(\mu_{l-\alpha+1}), \Delta(\Psi_2^*(\kappa^{-1}_\delta)) \big) \Big)_{\substack{1 \le \alpha \le l \\ 1 \le \delta \le k-i}} &
    \Big( \Lambda\Omega\big(\Psi_1(\mu_{l-\alpha+1}), \Delta(\Psi_1^*(\lambda_\gamma)) \big) \Big)_{\substack{1 \le \alpha \le l \\ 1 \le \gamma \le i}}
    \end{pmatrix}_{s\times k}, \\[12pt]
    \widetilde{\mathbf{B}}  &:= \begin{pmatrix}
    \left( (\Delta^*)^m \Big( \Delta(\Psi_2^*(\kappa^{-1}_\delta)) \Big) \right)_{\substack{0 \le m \le k-s-1 \\ 1 \le \delta \le k-i}} &
    \left( (\Delta^*)^m \Big( \Delta(\Psi_1^*(\lambda_\gamma)) \Big) \right)_{\substack{0 \le m \le k-s-1 \\ 1 \le \gamma \le i}}
    \end{pmatrix}_{(k-s)\times k}, \\[12pt]
    \widetilde{ {\mathbf{C}}}  &:= \begin{pmatrix}
    \left( (\Delta^*)^m \Big( \Delta(\Psi_2^*(\kappa^{-1}_\delta)) \Big) \right)_{\substack{0 \le m \le k-s-2 \\ 1 \le \delta \le k-i}} &
    \left( (\Delta^*)^m \Big( \Delta(\Psi_1^*(\lambda_\gamma)) \Big) \right)_{\substack{0 \le m \le k-s-2 \\ 1 \le \gamma \le i}}
    \end{pmatrix}_{(k-s-1)\times k},\\[12pt]
    \widetilde{\mathbf{D}}  &:= \begin{pmatrix}
    \Big( \Psi_2^*(n+1, \mathbf{t},\kappa^{-1}_\delta) \Big)_{\substack{  1 \\ 1 \le \delta \le k-i}} &
    \Big( \Psi_1^*(n+1,\mathbf{t}, \lambda_\gamma) \Big)_{\substack{  1 \\ 1 \le \gamma \le i}}
    \end{pmatrix}_{1\times k}.
\end{align*}
\end{small}
\end{theorem}
\begin{remark}
In this theorem, we apply four distinct vertex operators ($X_1^*, X_2^*, X_2, X_1$) to the mToda tau pair $(\tau_{0,n}, \tau_{1,n})$ in a successive, mixed, and asymmetric manner. The actions of these vertex operators on the tau functions are expressed in terms of generalized Wronskian determinants, which are known as the   addition formulas or generalized Fay identities \cite{Nakayashiki2026}. These determinant representations characterize the cross-flow interactions between $\mathbf{t}^{(1)}$ and $\mathbf{t}^{(2)}$. Furthermore, these representations play an important role in the connection between integrable systems and random matrix theory  \cite{Adler1999}, the determinant structures generated by vertex operators can serve as Christoffel-Darboux kernels and be utilized to  compute Fredholm determinants.
\end{remark}

\section{Transformations of the mToda   tau functions}\label{section2}
In this section, we investigate the relationship between the Toda and mToda tau  functions, and subsequently present the transformations for the mToda tau   pair, wave operators, and Lax operators. First, we provide the proof of Theorem \ref{Todaeigadjo}, demonstrating that both   of the mToda  tau  functions are two Toda  tau  functions. Based on this, we further prove Theorem \ref{mTodaeigentran}.
\subsection{Relationship between Toda and mToda tau functions}
Before proving Theorem \ref{Todaeigadjo}, we introduce the following lemma.
\begin{lemma}\label{DiffFerence}
The mToda wave functions $\Psi_a(n,\mathbf{t},z)$ and the mToda adjoint wave functions $\Psi_a(n,\mathbf{t},z)$ satisfy
\begin{align*}
&\Psi_1(n+1,\mathbf{t},z)-\Psi_1(n,\mathbf{t},z)
=
\frac{\tau_{0,n+1}(\mathbf{t})\tau_{1,n}(\mathbf{t}-[z^{-1}]_1)}
{\tau_{1,n}(\mathbf{t})\tau_{1,n+1}(\mathbf{t})}
z^{n+1}e^{\xi({\mathbf{t}}^{(1)},z)}, \\
&\Psi_1^*(n+1,\mathbf{t},z)-\Psi_1^*(n,\mathbf{t},z)
=-\frac{\tau_{0,n+1}(\mathbf{t}+[z^{-1}]_1)
\tau_{1,n}(\mathbf{t})}
{\tau_{0,n}(\mathbf{t})\tau_{0,n+1}(\mathbf{t})}
z^{-n}e^{-\xi({\mathbf{t}}^{(1)},z)}, \\
&\Psi_2(n+1,\mathbf{t},z)-\Psi_2(n,\mathbf{t},z)
=
-\frac{\tau_{1,n+1}(\mathbf{t}-[z]_2)\tau_{0,n+1}(\mathbf{t})}
{\tau_{1,n}(\mathbf{t})\tau_{1,n+1}(\mathbf{t})}
z^{n}e^{\xi({\mathbf{t}}^{(2)},z^{-1})}, \\
&\Psi_2^*(n+1,\mathbf{t},z)-\Psi_2^*(n,\mathbf{t},z)
=\frac{\tau_{1,n}(\mathbf{t})
\tau_{0,n}(\mathbf{t}+[z]_2)}
{\tau_{0,n}(\mathbf{t})\tau_{0,n+1}(\mathbf{t})}
z^{-n}e^{-\xi({\mathbf{t}}^{(2)},z^{-1})}.
\end{align*}
\end{lemma}
\begin{proof}
Let   $n-n'=1$, $\mathbf{t}-\mathbf{t}'=[z^{-1}]_1$ and  $n=n'$, $\mathbf{t}-\mathbf{t}'=[z]_2$ in \eqref{HirotamToda}, respectively, then we can  get
\begin{align*}
&-z\tau_{0,n}(\mathbf{t}+[z^{-1}]_1)
\tau_{1,n-1}(\mathbf{t})
+z\tau_{0,n}(\mathbf{t})
\tau_{1,n-1}(\mathbf{t}+[z^{-1}]_1)
=\tau_{1,n}(\mathbf{t}+[z^{-1}]_1)\tau_{0,n-1}(\mathbf{t}),\\
&\tau_{0,n}(\mathbf{t})
\tau_{1,n}(\mathbf{t}+[z]_2)
-z\tau_{0,n+1}(\mathbf{t})
\tau_{1,n-1}(\mathbf{t}+[z]_2)
=\tau_{1,n}(\mathbf{t})\tau_{0,n}(\mathbf{t}+[z]_2).
\end{align*}
The lemma follows directly from the above equations.
\end{proof}

We now {\bf prove Theorem \ref{Todaeigadjo}}. First,
apply  $\Delta$ to the mToda bilinear equation \eqref{mTodaWavbi}  with respect to $n$. By     Lemma \ref{DiffFerence},
we  can get
 \begin{align*}
&\oint_{C_{R}}\frac{dz}{2\pi {\rm i} }z^{n-n'}\tau_{1,n}(\mathbf{t}-[z^{-1}]_1)\tau_{1,n'}(\mathbf{t}'+[z^{-1}]_1)
e^{\xi({\mathbf{t}}^{(1)}-{\mathbf{t}}^{(1)\prime},z)}\nonumber\\
&\qquad=
\oint_{C_{r}}\frac{dz}{2\pi {\rm i} }z^{n-n'}\tau_{1,n+1}(\mathbf{t}-[z]_2)
\tau_{1,n'-1}(\mathbf{t}'+[z]_2)
e^{\xi({\mathbf{t}}^{(2)}-{\mathbf{t}}^{(2)\prime},z^{-1})},
\end{align*}
which coincides exactly with the Toda bilinear equation  \eqref{todabilitau}, demonstrating that $\tau_{1,n}(\mathbf{t})$ is a   Toda tau function. Similarly, applying $\Delta$ to the mToda bilinear equation \eqref{mTodaWavbi}    with respect to $n'$, we can know  that $\tau_{0,n}(\mathbf{t})$ is also a Toda tau function.

Thus, we can  define  the Toda wave functions $\psi^{[i]}_b$ and the Toda adjoint wave functions $\psi_b^{[i]*}$, by setting $\tau^{{\rm Toda}}_{n}=\tau_{i,n}$ with $i=0,1$ in \eqref{Todawave1} and \eqref{Todawave2}. It can be found that
\begin{align}
\pa_{t^{(a)}_k}\big(\psi_b^{[i]}\big)
=(\mathcal{B}_k^{(a)})^{[i]}\big(\psi_b^{[i]}\big),\quad
\pa_{t^{(a)}_k}\big(\psi_b^{[i]*}\big)
=-(\mathcal{B}_k^{(a)})^{[i]*}\big(\psi_b^{*[i]}\big),\quad  a,b=1,2. \label{Todawaveen}
\end{align}
Here   $(\mathcal{B}_k^{(a)})^{[i]}$ is given by \eqref{TodaLax} with the Toda Lax operators  $\mathcal{L}^{[i]}_1$ and $\mathcal{L}^{[i]}_2$ correspond to   the Toda tau function $\tau_{i,n}$ by \eqref{TodaLax1} and \eqref{TodaLax2}.
Further, defining
 \begin{align*}
q_n(\mathbf{t})=\tau_{1,n}(\mathbf{t})/\tau_{0,n}(\mathbf{t}),
\end{align*}
and substituting   into \eqref{HirotamToda},  we derive by  \eqref{Todawave1} and \eqref{Todawave2},
\begin{align}
		q_n( \mathbf{t})=&\oint_{C_R}\frac{{\rm d}z}{2\pi {\rm i}}
		z^{-2}q_{n'}(\mathbf{t}'+[z^{-1}]_1)
\psi_1^{*[0]}(n'-1,\mathbf{t}',z)
\psi_1^{[0]}(n,\mathbf{t},z)\nonumber\\
		&+\oint_{C_r}\frac{{\rm d}z}{2\pi {\rm i}}z^{-1}q_{n'-1}(\mathbf{t}'+[z]_2)\psi_2^{*[0]}(n'-1,\mathbf{t}',z)
\psi_2^{[0]}(n,\mathbf{t},z).\label{phires2}
		\end{align}
So applying   $\partial_{ {t}_k^{(a)}}$ to \eqref{phires2}, we can find by \eqref{Todawaveen} that $q_n( \mathbf{t})$  satisfies \eqref{Todaeig2}.   Thus  $q_n(\mathbf{t})$  is Toda eigenfunction  associated with $\tau_{0,n}(\mathbf{t})$. Similarly, if we define
$r_n(\mathbf{t})=\tau_{0,n+1}(\mathbf{t})/\tau_{1,n+1}(\mathbf{t})$,  we can know that $r_n(\mathbf{t})$
 is the Toda adjoint eigenfunction relative to  $\tau_{1,n}(\mathbf{t})$. {\bf This concludes the first part of Theorem \ref{Todaeigadjo}}.

{\bf For the second part of Theorem  \ref{Todaeigadjo}},   given     the Toda eigenfunction $q_n(\mathbf{t})$, we can get \eqref{qToda}.  (For more details, see Lemma 2 and Theorem 8 in \cite{Lui2024}).
If we define
$\tau_{0,n}(\mathbf{t}):=\tau^{{\rm Toda}}_n(\mathbf{t}),$  $\tau_{1,n}(\mathbf{t}):=q_n(\mathbf{t})\tau^{{\rm Toda}}_n(\mathbf{t}),$
 then we can find
the tau pair   $(\tau_{0,n}(\mathbf{t}),\tau_{1,n}(\mathbf{t})):=(\tau^{{\rm Toda}}_{n}(\mathbf{t}), q_n(\mathbf{t})\tau^{{\rm Toda}}_n(\mathbf{t}))$ satisfies the mToda bilinear equation \eqref{HirotamToda}. In this case,    $\tau_{0,n}(\mathbf{t})$ and $\tau_{1,n}(\mathbf{t})$ are two   Toda tau functions by the first part of Theorem \ref{Todaeigadjo}. Therefore,   $\tau_n^{[1]}(\mathbf{t}):=q_n(\mathbf{t})\tau^{{\rm Toda}}_n(\mathbf{t})$ is a new Toda tau function.

Similarly, for the Toda adjoint eigenfunction $r_n$   satisfying \eqref{rToda},
we can find that the tau pair $(\tau_{0,n}(\mathbf{t}),\tau_{1,n}(\mathbf{t})):=( r_{n-1} (\mathbf{t})\tau^{{\rm Toda}}_{n}(\mathbf{t}),\tau^{{\rm Toda}}_{n}(\mathbf{t}))$ is the mToda tau pair, and  $\tau_n^{[-1]}(\mathbf{t}):=r_{n-1} (\mathbf{t})\tau^{{\rm Toda}}_{n}(\mathbf{t})$ also is a new Toda tau function. {\bf This completes the proof of Theorem \ref{Todaeigadjo}}.
\subsection{Transformations of the mToda hierarchy in tau functions}
The primary objective of this subsection is to prove Theorem \ref{mTodaeigentran}. Before proving Theorem 2, let us recall that the mToda hierarchy is related to the Toda hierarchy via Miura transformations, as shown in the lemma below.
\begin{lemma}\label{muratras}\cite{Rui2024}
Given mToda objects: tau pair $(\tau_{0,n},\tau_{1,n})$,  Lax operators $L_a$   $(a=1,2)$, eigenfunction $\widetilde{q}_{n}(\mathbf{t})$, adjoint eigenfunction $\widetilde{r}_{n}(\mathbf{t})$, if set $q_n=\frac{\tau_{1,n}}{\tau_{0,n}}$,  $r_n=\frac{\tau_{0,n+1}}{\tau_{1,n+1}}$, $\mathcal{L}_a$,  $f_n$ and $g_n$ in the form of Table I
\begin{center}
\begin{tabular}{lll}
\multicolumn{3}{c}{Table I. Miura transformations: mToda $\rightarrow$ Toda}\\
\hline \hline
 &$\mathcal{T}_1=q_n$ &\quad$ \mathcal{T}_2=r^{-1}_n\Delta$ \ \\
\hline
{\rm Lax operators} &$\mathcal{L}_1=
{\mathcal{T}}_1\cdot{L}_1\cdot{\mathcal{T}}_1^{-1}$ &\quad $ \mathcal{L}_1=\mathcal{T}_2\cdot L_1\cdot \mathcal{T}_2^{-1}$\\
  &$\mathcal{L}_2=
{\mathcal{T}}_1\cdot{L}_2\cdot{\mathcal{T}}_1^{-1}$ &\quad $ \mathcal{L}_2=\mathcal{T}_2\cdot L_2\cdot \mathcal{T}_2^{-1}$\\
\hline
{\rm Eigenfunction}&${f_n}=\widetilde{q}_{n}
{q_{n}}$ & \quad$  {f_n}=\frac{\widetilde{q}_{n+1}-\widetilde{q}_{n}}{r_n}$\\
\hline
{\rm Adjoint eigenfunction}&$g_n=\frac{\widetilde{r}_{n+1}-\widetilde{r}_{n}} {q_{n}}
 $ &\quad $g_n=\frac{\widetilde{r}_{n+1}}{r_{n}}$\\
\hline
{\rm Tau function} &$\tau_{n}^{\rm Toda}={\rm const}\cdot\tau_{0,n}$& \quad$\tau_{n}^{\rm Toda}={\rm const}\cdot\tau_{1,n}$\\
\hline
\end{tabular}
\end{center}
 then $\mathcal{L}_1$ and $\mathcal{L}_2$ are Toda Lax operators,    $q_n$   and $f_n$ are Toda   eigenfunctions,  $r_n$ and $g_n$ are Toda adjoint eigenfunction.
 Here $\Delta^{-1}$ in $T_2^{-1}$ is expanded   in the form    $\sum_{ -\infty\ll i }a_i(n)\Lambda^{i}$ for $\mathcal{L}_1 =\mathcal{T}_{2}\cdot {L}_1 \cdot  \mathcal{T}_{2}^{-1}$,   while for $\mathcal{L}_2 =\mathcal{T}_{2}\cdot {L}_2 \cdot \mathcal{T}_{2}^{-1}$, it  is expressed in the form $\sum_{i\ll +\infty}a_i(n)\Lambda^i$.
\end{lemma}
\begin{lemma}\label{Todadarbo}\cite{Rui2024}
Given the Toda tau function $\tau^{{\rm Toda}}_{n}(\mathbf{t})$,   Toda Lax operators $\mathcal{L}_a$ for $a=1,2$,    Toda eigenfunctions $q_n$ and $f_{n}$,     Toda adjoint eigenfunctions $r_n$ and $g_{n}$, if we denote $\mathcal{L}^{[1]}_a$, ${f}_{n}^{[1]}$ and ${g}_{n}^{[1]}$ in the following Table II:
\begin{center}
\begin{tabular}{lll}
\multicolumn{3}{c}{Table II. Darboux transformations of the Toda hierarchy}\\
\hline \hline
&$T_D= q_{n+1}\cdot \Delta\cdot q_n^{-1}$ &$T_I= r^{-1}_{n-1}\cdot \Delta^{-1}\cdot r_n$   \\
\hline
{\rm Lax operators}&$\mathcal{L}_1^{[1]} =T_D\cdot \mathcal{L}_1 \cdot T_D^{-1}$&$\mathcal{L}_1^{[1]} = T_I\cdot \mathcal{L}_1\cdot  T_I^{-1} $\\
&$\mathcal{L}_2^{[1]} = T_D\cdot  \mathcal{L}_2 \cdot T_D^{-1} $&$\mathcal{L}_2^{[1]} = T_I\cdot \mathcal{L}_2 \cdot T_I^{-1}$\\
\hline
{\rm Eigenfunction}&${f}_{n}^{[1]} = T_D({q}_{n})({f}_{n})$ & $ {f}_{n}^{[1]} = T_I({r}_{n})({f}_{n})$\\
\hline
{\rm Adjoint eigenfunction}&${g}_{n}^{[1]} = (T_D^{-1})^*({q}_{n})({g}_{n})$ & $ {g}_{n}^{[1]} = (T_I^{-1})^*({r}_{n})({g}_{n})$\\
\hline
{\rm Tau function} &$(\tau^{{\rm Toda}}_{n})^{[1]} = q_n \tau^{{\rm Toda}}_{n}$  & $(\tau^{{\rm Toda}}_{n})^{[1]} = r_{n-1} \tau^{{\rm Toda}}_{n}$\\
\hline
\end{tabular}
\end{center}
then $\mathcal{L}^{[1]}_a$,  ${f}^{[1]}_{n}$ and ${g}^{[1]}_{n}$ are new the Toda Lax operators, the Toda wave operators, the Toda eigenfunction and the Toda adjoint eigenfunction respectively, corresponding to the new Toda tau function $(\tau^{{\rm Toda}}_{n})^{[1]}$.
\end{lemma}
After the preparation above, now we can {\bf prove Theorem \ref{mTodaeigentran}}. Here we only prove the $\textbf{Case 1}$, $\textbf{Case 2}$ can be proved analogously.
 Given the mToda tau pair $(\tau_{0,n},\tau_{1,n})$, we know by Theorem \ref{Todaeigadjo} that $\tau_{0,n}$ is a Toda tau function and $q_{n} := \tau_{1,n}/\tau_{0,n}$ is its associated Toda eigenfunction. By Lemma \ref{muratras}, we have  $f_{n}=\widetilde{q}_nq_n$, $f_{n}$ can be another Toda eigenfunction.   Therefore,   it can be found that
\begin{align*}
\tau_{0, n}^{[1]}
=f_{n}\tau_{0,n},\quad
\tau_{1, n}^{[1]}
=q_{n}^{[1]}f_{n}\tau_{0,n}=q_{n}^{[1]}\tau_{0, n}^{[1]},
\end{align*}
where $q_{n}^{[1]}=T_D(f_{n})(q_n)=\frac{q_{n+1}f_n-f_{n+1}q_n}{f_n}$ can be regarded as another Toda eigenfunction associated with $\tau_{0,n}^{[1]}$, according to Lemma \ref{Todadarbo}. From the second  part of Theorem \ref{Todaeigadjo}, $\tau_{i, n}^{[1]}$ $(i=0,1)$   also are new Toda functions, and the  tau pair $(\tau_{0,n}^{[1]}, \tau_{1,n}^{[1]})$ satisfy the mToda bilinear equation \eqref{HirotamToda}.  {\bf This concludes the first part of Theorem \ref{mTodaeigentran}}.

Next,{ \bf for the second part of Theorem  \ref{mTodaeigentran}}, we need the following lemma.
\begin{lemma}\label{Lem:eq}
   The eigenfunctions $\widetilde{q}_{n}$ and the adjoint eigenfunctions $\widetilde{r}_{n}$ of the mToda hierarchy satisfy
    \begin{align*}
        &\widetilde{q}_{n}(\mathbf{t})\Delta\big(\Psi_1^*(n,\mathbf{t},z)\big)
=\Delta\Big(
        \widetilde{q}_{n}(\mathbf{t}+[z^{-1}]_1)
        \Psi_1^*(n,\mathbf{t},z)\Big), \\
        &\widetilde{q}_{n}(\mathbf{t})\Delta\big(\Psi_2^*(n,\mathbf{t},z)\big)
=\Delta\Big(\widetilde{q}_{n-1}(\mathbf{t}+[z]_2)
        \Psi_2^*(n,\mathbf{t},z)\Big), \\
        &\Delta\big(\Psi_1(n,\mathbf{t},z)\big)\widetilde{r}_{n+1}(\mathbf{t})
=\Delta\Big(\widetilde{r}_{n}(\mathbf{t}-[z^{-1}]_1)
        \Psi_1(n,\mathbf{t},z)\Big), \\
        &\Delta\big(\Psi_2(n,\mathbf{t},z)\big)\widetilde{r}_{n+1}(\mathbf{t})
=\Delta\Big(\widetilde{r}_{n+1}(\mathbf{t}-[z]_2)
        \Psi_2(n,\mathbf{t},z)\Big).
    \end{align*}
\end{lemma}

\begin{proof}
Here we only prove the second identity, as the others can be proved similarly.
If we denote $q_n=\frac{\tau_{1,n}}{\tau_{0,n}}$, which is the  Toda eigenfunction  associated with $\tau_{0,n}$, we have by Lemma \ref{DiffFerence} and Theorem \ref{Todaeigadjo} that
\begin{align}
q_n(\mathbf{t})\psi^{[0]*}_2(n,\mathbf{t},z)
=\Delta\big(q_{n-1}(\mathbf{t}+[z]_2)
\psi^{[0]*}_2(n-1,\mathbf{t},z)\big).\label{eigenfunctionsep}
\end{align}
Here the Toda adjoint  wave function $\psi^{[0]*}_2=\frac{\tau_{0,n}(\mathbf{t}+[z]_2)}{\tau_{0,n+1}(\mathbf{t})}
e^{-\xi({\mathbf{t}}^{(2)},z^{-1})}z^{-n}$.
On the other hand, the Toda adjoint wave function $\psi^{[0]*}_2$ and the mToda adjoint wave function $\Psi^{*}_2$ satisfy the following relation (for more details, see Corollary 4.12 in \cite{Rui2024}):
\begin{align}
\psi_2^{[0]*}(n,\mathbf{t},z) =q^{-1}_n(\mathbf{t}) \Delta\big(\Psi_2^*(n,\mathbf{t},z)\big).\label{psi2}
\end{align}
Substituting this relation into \eqref{eigenfunctionsep}, we can get \begin{align}
\Psi_2^*(n,\mathbf{t},z)
=q_{n-1}(\mathbf{t}+[z]_2)\psi^{[0]*}_2(n-1,\mathbf{t},z).\label{Psi2}
\end{align}
Letting $f_{n}=\widetilde{q}_n q_n$, then by Lemma \ref{muratras}, $f_n$ is another Toda eigenfunction associated with $\tau_{0,n}$. In fact, according to Theorem \ref{Todaeigadjo}, any Toda eigenfunction satisfies the relation \eqref{eigenfunctionsep}. Therefore, $f_n$ satisfies
\begin{align}
f_n(\mathbf{t})\psi^{[0]*}_2(n,\mathbf{t},z)
=\Delta\big(f_{n-1}(\mathbf{t}+[z]_2)\psi^{[0]*}_2(n-1,\mathbf{t},z)\big).\label{eigenfunctionsep2}
\end{align}
Further, by \eqref{psi2} and \eqref{Psi2}, we can rewrite \eqref{eigenfunctionsep2}   as
\begin{align*}
\widetilde{q}_{n}(\mathbf{t}) \Delta\big(\Psi_2^*(n,\mathbf{t},z)\big) =&\Delta\big(\widetilde{q}_{n-1} (\mathbf{t}+[z]_2) q_{n-1}(\mathbf{t}+[z]_2)\psi^{[0]*}_2(n-1,\mathbf{t},z)\big)\nonumber\\
=& \Delta\Big(\widetilde{q}_{n-1}(\mathbf{t}+[z]_2) \Psi_2^*(n,\mathbf{t},z)\Big),
\end{align*}
which is just the second identity.
\end{proof}

Next, we continue to {\bf prove the second part of Theorem \ref{mTodaeigentran}}.
Let $\Psi_1^{[1]}$ and $\Psi_2^{[1]}$ be the  mToda wave functions associated with the mToda tau pair $(\tau_{0,n}^{[1]},\tau_{1,n}^{[1]})$   defined by \eqref{mToadwave1} and \eqref{mToadwave2}. Then by Lemma \ref{DiffFerence} and Lemma \ref{Lem:eq}, we have
  \begin{align*}
\Psi_1^{*[1]}(n,\mathbf{t}, z)
=&z\widetilde{q}_n(\mathbf{t})^{-1}\cdot \Omega\Big(\Delta\big( \widetilde{q}_n(\mathbf{t})\big),
\Psi_1^*(n+1,\mathbf{t}, z)\Big),\\
\Psi_2^{*[1]}(n,\mathbf{t}, z)
=&-\widetilde{q}_n(\mathbf{t})^{-1}\cdot \Omega\Big(\Delta \big(\widetilde{q}_n(\mathbf{t})\big),\Psi_2^*(n+1,\mathbf{t}, z)\Big),
	\end{align*}
where  $\Omega\Big(\Delta \big(\widetilde{q}_n(\mathbf{t})\big),\Psi_a^*(n+1,\mathbf{t}, z)\Big)$  $a=1,2$,   for more details see Appendix.

Let $S_1^{[1]}(n,\mathbf{t},\Lambda)$ and $S_2^{[1]}(n,\mathbf{t},\Lambda)$ be the mToda wave operators associated with the tau pair $(\tau_{0,n}^{[1]},\tau_{1,n}^{[1]})$ defined by \eqref{mTodaderess1}, such that
\begin{align*}
{\Psi}^{{[1]}*}_1(n,\mathbf{t},z)&=-e^{-\xi(\mathbf{t}^{(1)},z)}\cdot
\iota_{\La}\Delta^{-1}\big((S^{[1]*}_1)^{-1}\big)
(z^{-n}),\\
{\Psi}^{{[1]}*}_2(n,\mathbf{t},z)&=e^{-\xi(\mathbf{t}^{(2)},z^{-1})}\cdot
\iota_{\La^{-1}}\Delta^{-1}\big((S^{[1]*}_2)^{-1}\big)
(z^{-n}).
\end{align*}
Then, we find that
\begin{align*}
{S}_1^{[1]} =T_1 \cdot{S_1}\cdot\Lambda^{-1},\quad {S}_2^{[1]} =-T_1 \cdot{S_2},\quad T_1=\frac{\widetilde{q}_n\widetilde{q}_{n+1}}{\widetilde{q}_n-\widetilde{q}_{n+1}}
\cdot\Delta\cdot\widetilde{q}_n^{-1},
\end{align*}
which are the new mToda wave operators satisfying \eqref{mTodaSx1} and \eqref{mTodaSx2}.
Furthermore, if we denote $ L_1^{[1]}=S^{[1]}_1\cdot \Lambda \cdot (S_1^{[1]})^{-1}$ and $ L_2^{[1]}=S^{[1]}_2\cdot \Lambda^{-1} \cdot (S_2^{[1]})^{-1}$, we have
\begin{align*}
L_1^{[1]}=T_1\cdot  L_1 \cdot  T_1^{-1},\quad L_2^{[1]}=T_1\cdot  L_2 \cdot T_1^{-1},
\end{align*}
which are the new mToda Lax operators and satisfy the Lax equations:
\begin{eqnarray*}
\pa_{t_k^{(a)}}L_b^{[1]}=[B_k^{{(a)[1]}},L^{[1]}_b], \quad B_k^{(1){[1]}}= ( L^{[1]k}_1) _{\Delta,\geq1},\quad
		B_k^{{(2)[1]}}= ( L^{[1]k}_2)_{\Delta^{*},\geq1}, \quad a,b=1,2.
\end{eqnarray*}
 {\bf This concludes the second part of Theorem \ref{mTodaeigentran}}.

Further, given another mToda eigenfunction $\widetilde{f}_n$, setting $\widetilde{f}^{[1]}_n=T_b(\widetilde{f}_n)$, we have
\begin{align*}
\partial_{t_k^{(a)}}\widetilde{f}^{[1]}_n= \Big((\partial_{t_k^{(a)}}T_b)\cdot T_b^{-1} + T_b\cdot  B_k^{(a)}\cdot T_b^{-1}\Big)(\widetilde{f}^{[1]}_n),\quad a,b=1,2.\end{align*}
On the other hand, according to the second part of Theorem 2, $S_1^{[1]}$ and $S_2^{[1]}$ are the new mToda wave operators satisfying   \eqref{mTodaSx1} and \eqref{mTodaSx2}, implying that
\begin{align*}(\partial_{t_k^{(a)}}T_b)\cdot T_b^{-1} + T_b \cdot B_k^{(a)}\cdot T_b^{-1} = B_k^{(a)[1]}.\end{align*}
Therefore, we immediately get $\partial_{t_k^{(a)}}\widetilde{f}^{[1]}_n = B_k^{[1](a)}(\widetilde{f}^{[1]}_n)$, which implies that $\widetilde{f}^{[1]}_n = T_b(\widetilde{f}_n)$ is the   mToda eigenfunction corresponding to $(\tau^{[1]}_{0,n}, \tau^{[1]}_{1,n})$ in $\textbf{Case b}$. Similarly, $\widetilde{g}^{[1]}_n = \big(\Delta^{-1}(T_b^*)^{-1}\Delta\big)(\widetilde{g}_n)$ acts as the new mToda adjoint eigenfunction  corresponding to $(\tau^{[1]}_{0,n}, \tau^{[1]}_{1,n})$ in $\textbf{Case b}$ for $b=1,2$.
{\bf This completes the proof of Theorem \ref{mTodaeigentran}}.

If $T_{1}$ is applied to the mToda tau pair $(\tau_{0, n},\tau_{1,n})$ using the eigenfunction $\widetilde{q}_n$, followed by $T_{2}$ using the adjoint eigenfunction $\widetilde{r}^{[1]}_n$, we have:
\begin{align*}
 (\tau_{0, n} ,\tau_{1,n} )\xrightarrow{T_{1}}(\tau^{[1]}_{0,n} , \tau^{[1]}_{1,n} )\xrightarrow{T_{2}}(\tau^{\{1\}}_{0,n} , \tau^{\{1\}}_{1,n}),
\end{align*}
which implies that the second tau pair is given by
\begin{align*}
\tau_{0,n}^{[1]} =\widetilde{q}_n \tau_{1,n}, \quad \tau_{1,n}^{[1]} =\dfrac{ (\widetilde{q}_{n}-\widetilde{q}_{n+1} )\tau_{1,n} \tau_{1,n+1} } {\tau_{0,n+1} },
\end{align*}
and the final tau pair is
\begin{align*}
\tau_{0,n}^{\{1\}} = \frac{(\widetilde{r}^{[1]}_{n}-\widetilde{r}^{[1]}_{n-1})
\tau^{[1]}_{0,n}\tau^{[1]}_{0,n-1}}{\tau^{[1]}_{1,n-1}}, \quad \tau_{1,n}^{\{1\}} = \widetilde{r}^{[1]}_n \tau^{[1]}_{0,n},
\end{align*}
where $\widetilde{r}^{[1]}_n = \Delta^{-1} T_{1}(\widetilde{q}_n)(\Delta(\widetilde{r}_n)) = \frac{\widetilde{r}_n\widetilde{q}_n -\Omega\big(\widetilde{q}_n,\Delta(\widetilde{r}_n)\big)}{\widetilde{q}_n}$ is the mToda adjoint eigenfunction associated with the tau pair $(\tau^{[1]}_{0,n}, \tau^{[1]}_{1,n})$ by the third part of Theorem \ref{mTodaeigentran}.
Then,   we obtain:
\begin{align*}
\Big(\tau^{\{1\}}_{0,n},\tau^{\{1\}}_{1,n}\Big) =
\left(-\Omega\Big(\widetilde{q}_n, \Delta(\widetilde{r}_n)\Big)\tau_{0,n}, \Omega\Big( \Delta(\widetilde{q}_n), \widetilde{r}_{n+1}\Big)\tau_{1,n}\right).
\end{align*}
From the first part of Theorem \ref{mTodaeigentran}, the tau pair $(\tau^{\{1\}}_{0,n},\tau^{\{1\}}_{1,n})$ is also the mToda tau pair, satisfying the mToda bilinear equation \eqref{HirotamToda}.
Which is just  {\bf Corollary \ref{sepmToda}}.

\section{The   symmetries   of the mToda hierarchy}\label{section3}
In this section, we mainly investigate the symmetries  of the mToda hierarchy. First, we introduce the squared eigenfunction symmetries. Subsequently, by applying vertex operators to the tau functions, we present the mToda ASvM formula in Theorem \ref{ASvM}.
\subsection{Squared eigenfunction symmetries}
In this subsection, we define the squared eigenfunction symmetries for the mToda hierarchy and demonstrate that they commute with the time flows. We introduce the mToda additional flows $\partial_{\tilde{q}_{n},\tilde{r}_{n}}^{*}$ by
\begin{align*}
&\pa^{*}_{\widetilde{q}_{n},\widetilde{r}_{n}}S_1(n,\mathbf{t},\Lambda)
:=\widetilde{q}_{n}(\mathbf{t})
\cdot\iota_{\Lambda^{-1}}\Delta^{-1}
\cdot\widetilde{r}_{n+1}(\mathbf{t})\cdot\Delta \cdot S_1(n,\mathbf{t},\Lambda), \\ &\pa^{*}_{\widetilde{q}_{n},\widetilde{r}_{n}}S_2(n,\mathbf{t},\Lambda)
:=\widetilde{q}_{n}(\mathbf{t})
\cdot\iota_{\Lambda}\Delta^{-1}\cdot\widetilde{r}_{n+1}(\mathbf{t})\cdot\Delta \cdot S_2(n,\mathbf{t},\Lambda),
\end{align*}
where $\widetilde{q}_{n}$  and $\widetilde{r}_{n}$ are the mToda eigenfunction and  adjoint eigenfunction, respectively, satisfying \eqref{eigenfunction} and \eqref{adjointeigenfunction}.
Equivalently,   the additional flows $\pa^{*}_{\widetilde{q}_{n},\widetilde{r}_{n}}$ on   mToda Lax operators $L_1$  and $L_2$ by
\begin{align}
\pa^{*}_{\widetilde{q}_{n},\widetilde{r}_{n}}L_1=[\widetilde{q}_{n}\cdot
\iota_{\Lambda^{-1}}\Delta^{-1}\cdot\widetilde{r}_{n+1}\cdot \Delta,L_1],\quad \pa^{*}_{\widetilde{q}_{n},\widetilde{r}_{n}}L_2=[\widetilde{q}_{n}\cdot
\iota_{\Lambda}\Delta^{-1}\cdot\widetilde{r}_{n+1}\cdot \Delta,L_2].\label{symmetryLax}
\end{align}
\begin{proposition}\label{pro1}
The additional flows $ \pa^{*}_{\widetilde{q}_{n},\widetilde{r}_{n}}$ commute with  the time flows of the  mToda hierarchy, i.e.,
\begin{align*}
\big[\pa^{*}_{\widetilde{q}_{n},\widetilde{r}_{n}},\pa_{{t}^{(c)}_k}\big]=0,\quad c=1,2,
\end{align*}
which means $\pa^{*}_{\widetilde{q}_{n},\widetilde{r}_{n}}$ is just the  symmetries of mToda hierarchy, and is called the {\bf squared eigenfunction symmetries}.
\end{proposition}
\begin{proof}
Here we only prove the case for $c=1$, as the proof for $c=2$ follows similarly. If we denote   $Y_{ab} = \widetilde{q}_{a,n}\cdot \iota_{\Lambda^{-1}}\Delta^{-1}\cdot \widetilde{r}_{b, n+1}\cdot \Delta$, then   $\pa^{*}_{\widetilde{q}_{n},\widetilde{r}_{n}} S_1 = Y_{ab}S_1$.
According to $\pa_{{t}^{(1)}_k} S_1 = - (L_1^k)_{\Delta,\leq 0} S_1$,  we have:
\begin{align*}
\big[\pa^{*}_{\widetilde{q}_{n},\widetilde{r}_{n}}, \pa_{{t}^{(1)}_k}\big] S_1
&= - [Y_{ab}, L_1^k]_{\Delta,\leq 0} S_1 - (L_1^k)_{\Delta,\leq 0} Y_{ab} S_1 - \big( \pa_{{t}^{(1)}_k} Y_{ab} \big) S_1 + Y_{ab} (L_1^k)_{\Delta,\leq 0} S_1\\
&= \Big( - [Y_{ab}, L_1^k]_{\Delta,\leq 0} - [(L_1^k)_{\Delta,\geq 1}, Y_{ab}]_{\Delta,\leq 0}  + [Y_{ab}, (L_1^k)_{\Delta,\leq 0}] \Big) S_1 \\
&=\Big( - [Y_{ab}, (L_1^k)_{\Delta,\leq 0}]_{\Delta,\leq 0}+ [Y_{ab}, (L_1^k)_{\Delta,\leq 0}]_{\Delta,\leq 0} \Big)S_1=0,
\end{align*}
where for the second identity, we have used the fact that the difference operator $A=\sum_{i\in\mathbb{Z}}a_i\Delta^i$ and function $f(n,\mathbf{t})$ satisfy \cite{Song2023}
\begin{align*}
&A_{\Delta,\leq 0}=(A\Delta^{-1})_{\Delta,<0}\cdot\Delta,\quad  (A_{\Delta,\geq 1}f\Delta^{-1})_{\Delta,<0} = A_{\Delta,\geq 1}(f)\cdot\Delta^{-1},\\
&
(\Delta^{-1}f A_{\Delta,\geq 1})_{\Delta,<0} = \Delta^{-1}\cdot(A_{\Delta,\geq 1})^*(f).
\end{align*}
The case for $S_2$ can be proved similarly.
This completes the proof.
\end{proof}
\noindent{\bf Example:}
According to \eqref{Laxoperadj} and \eqref{symmetryLax}, the actions of the mToda squared eigenfunction symmetry flows $\pa^{*}_{\widetilde{q}_{n},\widetilde{r}_{n}}$ on the first few coefficients of $L_1$ and $L_2$ are  given by
\begin{align*}
&\pa^{*}_{\widetilde{q}_{n},\widetilde{r}_{n}}\beta_1(n,\mathbf{t})= -  \Delta(\widetilde{q}_{n}\widetilde{r}_{n}),\quad
\pa^{*}_{\widetilde{q}_{n},\widetilde{r}_{n}} u_{0,n} = \Delta(v_{n-1} \widetilde{q}_n \Delta(\widetilde{r}_{n-1})),\\
&\pa^{*}_{\widetilde{q}_{n},\widetilde{r}_{n}} u_{1,n}= u_{1,n} \Delta(\widetilde{q}_{n-1}\widetilde{r}_{n-1}) + \widetilde{q}_n\Delta(\widetilde{r}_{n-1})\Delta(u_{0,n-1}) + v_n\widetilde{q}_{n+1}\Delta(\widetilde{r}_{n-1}) - v_{n-2}\widetilde{q}_n\Delta(\widetilde{r}_{n-2}),
\end{align*}
and
\begin{align*}
&\pa^{*}_{\widetilde{q}_{n},\widetilde{r}_{n}}\beta_2(n,\mathbf{t})= \Delta(\widetilde{q}_{n-1}\widetilde{r}_{n}),\quad
\pa^{*}_{\widetilde{q}_{n},\widetilde{r}_{n}} \bar{u}_{0,n}= \Delta(\bar{v}_n \widetilde{q}_{n-1} \Delta(\widetilde{r}_{n})),\\
&\pa^{*}_{\widetilde{q}_{n},\widetilde{r}_{n}} \bar{u}_{1,n}= -\bar{u}_{1,n}\Delta(\widetilde{q}_n \widetilde{r}_{n+1}) + \widetilde{q}_n\Delta (\widetilde{r}_{n+1})\Delta(\bar{u}_{0,n}) + \bar{v}_{n+2}\widetilde{q}_n\Delta(\widetilde{r}_{n+2}) - \bar{v}_n\widetilde{q}_{n-1}\Delta(\widetilde{r}_{n+1}),
\end{align*}
where $v_n = e^{\beta_1(n,\mathbf{t})}$ and $\bar{v}_n = e^{\beta_2(n,\mathbf{t})}$.
\subsection{The   Adler-Shiota-van Moerbeke   formula}
After the preparations above, in this subsection we continue to prove {\bf Theorem \ref{ASvM}}, which establishes the mToda Adler-Shiota-van Moerbeke (ASvM) formula connecting the additional symmetries on the wave functions with the Sato--B\"{a}cklund transformations of the tau functions.

 Firstly, according to definition of the mToda wave functions  and   adjoint wave functions  with respect to the tau pair $(\tau_{0,n} ,\tau_{1,n} )$ in \eqref{mToadwave1} and \eqref{mToadwave2},
the actions of $\mathbb{X}_{ab}(n,\mathbf{t},\lambda,\mu)$ on the mToda tau pair $(\tau_{0,n}(\mathbf{t}),\tau_{1,n}(\mathbf{t}))$ are given by
\begin{align*}
&\frac{\mathbb{X}_{a1}(n,\mathbf{t},\lambda,\mu)\big(\tau_{0,n}(\mathbf{t})\big)}
{\tau_{0,n}(\mathbf{t})}
=(-1)^{\delta_{a,1}}\Psi_a(n,\mathbf{t}+[\lambda^{-1}]_1,\mu^{\eta_a})
\Psi^*_1(n,\mathbf{t},\lambda), \\
&\frac{\mathbb{X}_{a2}(n,\mathbf{t},\lambda,\mu)\big(\tau_{0,n}(\mathbf{t})\big)}
{\tau_{0,n}(\mathbf{t})}
=(-1)^{\delta_{a,2}}\Psi_a(n-1,\mathbf{t}+[\lambda^{-1}]_2,\mu^{\eta_a})
\Psi^*_2(n,\mathbf{t},\lambda^{-1}),\\
&\frac{\mathbb{X}_{1b}(n,\mathbf{t},\lambda,\mu)\big(\tau_{1,n}(\mathbf{t})\big)}
{\tau_{1,n}(\mathbf{t})}
=\frac{\lambda^{\delta_{b,1}}}{\mu}\Psi_1(n,\mathbf{t},\mu)
\Psi^*_b(n,\mathbf{t}-[\mu^{-1}]_1,\lambda^{\eta_b}),\\
&\frac{\mathbb{X}_{2b}(n,\mathbf{t},\lambda,\mu)\big(\tau_{1,n}(\mathbf{t})\big)}
{\tau_{1,n}(\mathbf{t})}
={\lambda^{\delta_{b,1}}}\Psi_2(n,\mathbf{t},\mu^{-1})
\Psi^*_b(n+1,\mathbf{t}-[\mu^{-1}]_2,\lambda^{\eta_b}).
\end{align*}
If we set $\widetilde{q}_n(\mathbf{t})=\Psi_a(n,\mathbf{t},\mu^{\eta_{a}})$, $\widetilde{r}_n(\mathbf{t})=\Psi_b^*(n,\mathbf{t},\lambda^{\eta_b})$ in Lemma \ref{Lem:eq}, and
by above equations, then we have
\begin{align}
&\frac{\mathbb{X}_{ab}(n,\mathbf{t},\lambda,\mu)
\big(\tau_{0,n}(\mathbf{t})\big)}{\tau_{0,n}(\mathbf{t})}
=(-1)^{\delta_{a,b}}\Omega\Big(\Psi_a(n,\mathbf{t},\mu^{\eta_a}),
\Delta(\Psi_b^*(n,\mathbf{t},\lambda^{\eta_b}))\Big),\label{Xab1}\\
&\frac{
\mathbb{X}_{ab}(n,\mathbf{t},\lambda,\mu)
\big(\tau_{1,n}(\mathbf{t})\big) }{\tau_{1,n}(\mathbf{t})}
=\frac{\lambda^{\delta_{b,1}}}{\mu^{\delta_{a,1}}}
\Omega\Big(\Delta(\Psi_a(n,\mathbf{t},\mu^{\eta_a})),
\Psi_b^*(n+1,\mathbf{t},\lambda^{\eta_b})\Big),\quad a,b=1,2.\label{Xab2}
\end{align}
Note that   $\Delta^{-1}(f_n)$ can be up to adding arbitrary   constant. Therefore, we obtain
\begin{align*}
1 + C_{ab}\frac{\mathbb{X}_{ab}(n,\mathbf{t},\lambda,\mu)
\big(\tau_{0,n}(\mathbf{t})\big)}{\tau_{0,n}(\mathbf{t})}
&=   C_{ab}(-1)^{\delta_{a,b}}\Omega\Big(\Psi_a(n,\mathbf{t},\mu^{\eta_a}),
\Delta\big(\Psi_b^*(n,\mathbf{t},\lambda^{\eta_b})\big)\Big), \\
1 + C_{ab}\frac{
\mathbb{X}_{ab}(n,\mathbf{t},\lambda,\mu)
\big(\tau_{1,n}(\mathbf{t})\big) }{\tau_{1,n}(\mathbf{t})}
&=   C_{ab}\frac{\lambda^{\delta_{b,1}}}{\mu^{\delta_{a,1}}}
\Omega\Big(\Delta\big(\Psi_a(n,\mathbf{t},\mu^{\eta_a})\big),
\Psi_b^*(n+1,\mathbf{t},\lambda^{\eta_b})\Big).
\end{align*}
Further, if we define
 \begin{align}
\tilde{\tau}_{i,n}^{ab}=\tau_{i,n}+C_{ab}\mathbb{X}_{ab}
(n,\mathbf{t},\lambda,\mu)\tau_{i,n},\quad i=0,1,\quad a,b=1,2,\label{satob}
 \end{align}
 then by Corollary \ref{sepmToda},  $(\tilde{\tau}_{0,n}^{ab},\tilde{\tau}_{1,n}^{ab})$ can be regarded as a new  mToda tau pair, satisfying the mToda bilinear equation \eqref{HirotamToda}. Here, \eqref{satob} are also called the Sato--B\"{a}cklund transformations of the mToda tau functions. {\bf This concludes the first part of Theorem \ref{ASvM}}.

Next, { \bf for the second part of Theorem  \ref{ASvM}}, we need the following lemma. \begin{lemma}\cite{Wang2025}\label{qrrelation}
The mToda eigenfunction  $\widetilde{q}_{n}$ and the mToda adjoint eigenfunction $\widetilde{r}_{n}$    satisfy the following relations:
\begin{align*}
\Delta\Big(\widetilde{q}_{n}(\mathbf{t})
\widetilde{r}_{n}(\mathbf{t}-[z^{-1}]_1)\Big)&=
\widetilde{q}_{n}(\mathbf{t}-[z^{-1}]_1)\Delta
\Big(\widetilde{r}_{n}(\mathbf{t}-[z^{-1}]_1)\Big)
+\Delta\Big(\widetilde{q}_{n}(\mathbf{t})\Big)
\widetilde{r}_{n+1}(\mathbf{t}),\\
\Delta\Big(\widetilde{q}_{n}(\mathbf{t})
\widetilde{r}_{n+1}(\mathbf{t}-[z]_2)\Big)&=
\widetilde{q}_{n+1}(\mathbf{t}-[z]_2)\Delta
\Big(\widetilde{r}_{n+1}(\mathbf{t}-[z]_2)\Big)
+\Delta\Big(\widetilde{q}_{n}(\mathbf{t})\Big)
\widetilde{r}_{n+1}(\mathbf{t}).
\end{align*}
\end{lemma}
Next, we continue to {\bf prove the second part of Theorem \ref{ASvM}}. Let us set $\widetilde{q}_n(\mathbf{t})=\Psi_a(n,\mathbf{t},\mu^{\eta_{a}})$ and $\widetilde{r}_n(\mathbf{t})=\Psi_b^*(n,\mathbf{t},\lambda^{\eta_b})$ in Lemma \ref{qrrelation}. Denoting ${\pa}_{{ \lambda,\mu}}^{ab}:=\pa^{*}_{\widetilde{q}_{n},\widetilde{r}_{n}}$, we obtain from \eqref{Xab1} and \eqref{Xab2}:
\begin{align}
\frac{{\pa}_{{ \lambda,\mu}}^{ab}(\Psi_1(n,\mathbf{t},z))}{\Psi_1(n,\mathbf{t},z)}
&=(-1)^{\delta_{a,b}}e^{-\xi(\widetilde{\partial}^{(1)},z^{-1})}
\frac{\mathbb{X}_{ab}(n,\mathbf{t},\lambda,\mu)
\tau_{0,n }(\mathbf{t})}{\tau_{0,n }(\mathbf{t})}
+\frac{\mu^{\delta_{a,1}}}{\lambda^{\delta_{b,1}}}\frac{
\mathbb{X}_{ab}(n,\mathbf{t},\lambda,\mu)
\tau_{1,n}(\mathbf{t}) }{\tau_{1,n}(\mathbf{t})},\label{ASvM1}\\
\frac{{\pa}_{{ \lambda,\mu}}^{ab}(\Psi_2(n,\mathbf{t},z))}{\Psi_2(n,\mathbf{t},z)}
&=(-1)^{\delta_{a,b}}e^{-\xi(\widetilde{\partial}^{(2)},z)}
\frac{\mathbb{X}_{ab}(n+1,\mathbf{t},\lambda,\mu)
\tau_{0,n+1}(\mathbf{t})}{\tau_{0,n+1}(\mathbf{t})}
+\frac{\mu^{\delta_{a,1}}}{\lambda^{\delta_{b,1}}}\frac{
\mathbb{X}_{ab}(n,\mathbf{t},\lambda,\mu)
\tau_{1,n}(\mathbf{t}) }{\tau_{1,n}(\mathbf{t})}.
\label{ASvM2}
\end{align}
 {\bf This concludes the second part of Theorem \ref{ASvM}}.

Next,  by   \eqref{mToadwave1}, we have:
\begin{align*}
&\frac{{\pa}_{{ \lambda,\mu}}^{ab}(\Psi_1(n,\mathbf{t},z))}{\Psi_1(n,\mathbf{t},z)}
=e^{-\xi(\widetilde{\partial}^{(1)},z^{-1})}
{\pa}_{{ \lambda,\mu}}^{ab}{\rm log }\tau_{0,n}(\mathbf{t})-
{\pa}_{{ \lambda,\mu}}^{ab}{\rm log }\tau_{1,n}(\mathbf{t}),\\
&\frac{{\pa}_{{ \lambda,\mu}}^{ab}(\Psi_2(n,\mathbf{t},z))}{\Psi_2(n,\mathbf{t},z)}
=e^{-\xi(\widetilde{\partial}^{(2)},z )}
{\pa}_{{ \lambda,\mu}}^{ab}{\rm log }\tau_{0,n+1}(\mathbf{t})-
{\pa}_{{ \lambda,\mu}}^{ab}{\rm log }\tau_{1,n}(\mathbf{t}).
\end{align*}
Further,  by \eqref{ASvM1} and \eqref{ASvM2}, we can get
\begin{align*}
\frac{\pa_{\lambda,\mu}^{ab} \tau_{0,n}}{\tau_{0,n}} = (-1)^{\delta_{a,b}} \frac{\mathbb{X}_{ab}(n,\mathbf{t},\lambda,\mu) \tau_{0,n}}{\tau_{0,n}}, \quad
\frac{\pa_{\lambda,\mu}^{ab} \tau_{1,n}}{\tau_{1,n}} = - \frac{\mu^{\delta_{a,1}}}{\lambda^{\delta_{b,1}}} \frac{\mathbb{X}_{ab}(n,\mathbf{t},\lambda,\mu) \tau_{1,n}}{\tau_{1,n}}.
\end{align*}
Multiplying both sides by their corresponding tau functions, we can get
\begin{align*}
&{\pa}_{\lambda,\mu}^{ab}\tau_{0,n}=(-1)^{\delta_{a,b}}
\mathbb{X}_{ab}
(n,\mathbf{t},\lambda,\mu)\tau_{0,n},\quad
{\pa}_{ \lambda,\mu}^{ab}\tau_{1,n}=-
\frac{\mu^{\delta_{a,1}}}{\lambda^{\delta_{b,1}}}
\mathbb{X}_{ab}
(n,\mathbf{t},\lambda,\mu)\tau_{1,n}.
\end{align*}
{\bf This completes the proof of Theorem \ref{ASvM}}.
\section{The     general addition formulas of the mToda hierarchy}\label{section4}
In this section, we first investigate the multi-step   transformations of the mToda tau functions. After that,   we present the proof of Theorem \ref{mTodahigerfayide} and derive the mToda generalized addition formulas, also known as the generalized Fay identities.
\subsection{Multi-step transformations of the tau functions}
In this subsection, we  first  investigate the multi-step   transformations of the Toda tau functions.  Subsequently, by using the relationship between the Toda and mToda tau functions, we present the multi-step transformations for the mToda hierarchy.
\begin{lemma}\label{TodatauDX}\cite{Wang2025JGP}
Given the Toda eigenfunctions $q_n, q_{i,n}$ ($1\leq i\leq s$) and the Toda adjoint eigenfunctions $r_n, r_{j,n}$ ($1\leq j\leq k$) with respect to the tau function $\tau^{{\rm Toda}}_n$, if we define $q_n^{[s+k]}$, $r_n^{[s+k]}$, and $(\tau_n^{\rm Toda})^{[s+k]}$ as follows:
\begin{itemize}
	\item $s\geq k$,
	\begin{align*} &q_n^{[s+k]}=\frac{{IW}_{k,s+1}(r_{k},\dots,r_{1};q_{1},\dots,q_{s},q)}
{IW_{k,s}(r_{k},\dots,r_{1};q_{1},\dots,q_{s})},\nonumber\\
		&r_{n-1}^{[s+k]}=(-1)^{s}\frac{ IW_{k+1,s}(r, r_{k},\dots,r_{1};q_{1},\dots,q_{s}) }
{ IW_{k,s}(r_{k},\dots,r_{1};q_{1},\dots,q_{s}) },\nonumber\\
		&(\tau_n^{\rm Toda})^{[s+k]}=(-1)^{sk}{IW}_{k,s}(r_{k},\dots,r_{1};q_{1},\dots,q_{s})\tau_n^{\rm Toda},
	\end{align*}
\end{itemize}
\begin{itemize}
	\item $s<k$,
	\begin{align*}
		&q_{n+1}^{[s+k]}=(-1)^{k}\frac{  IW^{*}_{s+1,k}(q,q_{s},\dots,q_{1};r_{1},\dots,r_{k}) }{  IW^{*}_{s,k}( q_{s},\dots,q_{1};r_{1},\dots,r_{k})  },\nonumber\\
		&r^{[s+k]}=\frac{IW_{s,k+1}^{*}(q_{s},\dots,q_{1};r_{1},\dots,r_{k},r)}
{IW_{s,k}^{*}( q_{s},\dots,q_{1};r_{1},\dots,r_{k})},\nonumber\\
		&(\tau_{n+1}^{\rm Toda})^{[s+k]}= (-1)^{sk} IW_{s,k}^{*}
( q_{s},\dots,q_{1};r_{1},\dots,r_{k}) \tau_{n+1}^{\rm Toda},
	\end{align*}
\end{itemize}
then $q^{[s+k]}$ is also   Toda eigenfunction, $r^{[s+k]}$ is  Toda adjoint eigenfunction, and $(\tau_n^{\rm Toda})^{[s+k]}$ is a new Toda tau function, satisfying the Toda bilinear equation \eqref{todabilitau}.
Here    we denote $q_i:= q_{i,n}$, $r_j:= r_{j,n}$, $q :=q_n$, and $r:= r_n$.  $IW_{k,s}$ and $IW^{*}_{s,k}$ are the generalized discrete Wronskian determinants defined as
\begin{align*}
	&IW_{k,s}(g_k,\dots,g_1;f_1,\dots,f_{s}) =
	\begin{vmatrix}
		\Big( \Delta^{-1}(f_j g_{k-i+1}) \Big)_{1 \leq i \leq k, \, 1 \leq j \leq s} \\[10pt]
		\Big( \Delta^{i-1}(f_j) \Big)_{1 \leq i \leq s-k, \, 1 \leq j \leq s}
	\end{vmatrix}, \quad (s \geq k), \\[15pt]
	&IW^{*}_{s,k}(f_s,\dots,f_1;g_1,\dots,g_k) =
	\begin{vmatrix}
		\Big( (\Delta^*)^{-1}(g_j f_{s-i+1}) \Big)_{1 \leq i \leq s, \, 1 \leq j \leq k} \\[10pt]
		\Big( (\Delta^*)^{i-1}(g_j) \Big)_{1 \leq i \leq k-s, \, 1 \leq j \leq k}
	\end{vmatrix}, \quad (s < k).
\end{align*}
\end{lemma}
\begin{proof}
This lemma can be proved by mathematical induction using the one-step Darboux transformations given in Lemma \ref{Todadarbo}, from which the multi-step determinant representations can be   obtained. Furthermore, by   Theorem 1, we can find that $(\tau_n^{\mathrm{Toda}})^{[s+k]}$ is a new Toda tau function, satisfying the Toda bilinear equation \eqref{todabilitau}.
\end{proof}
\begin{proposition}\label{mTodatautra}
Given the mToda eigenfunctions $\widetilde{q}_n, \widetilde{q}_{i,n}$ ($1\leq i\leq s$), the mToda adjoint eigenfunctions $\widetilde{r}_n, \widetilde{r}_{j,n}$ ($1\leq j\leq k$), and the mToda tau pair $(\tau_{0, n} ,\tau_{1,n} )$, we have that $\tau_{0,n}^{[s+k]}$ and $\tau_{1,n}^{[s+k]}$ satisfy:
\begin{itemize}
  \item  $s\geq k$
\begin{align*}
    \tau_{0,n}^{[s+k]} &= (-1)^{s k}  {IW}_{k,s}(\Delta(\widetilde{r}_k),\dots,\Delta(\widetilde{r}_1); \widetilde{q}_1,\dots,\widetilde{q}_s)\cdot\left( \prod_{m=0}^{s-k-1} \frac{\tau_{1,n+m}}{\tau_{0,n+m}} \right)\cdot \tau_{0,n},   \\
    \tau_{1,n}^{[s+k]} &= (-1)^{s k}  {IW}_{k, s+1}(\Delta(\widetilde{r}_k),\dots,\Delta(\widetilde{r}_1); \widetilde{q}_1,\dots,\widetilde{q}_s,1)\cdot\left( \prod_{m=0}^{s-k} \frac{\tau_{1,n+m}}{\tau_{0,n+m}} \right) \cdot \tau_{0,n},
\end{align*}
\item   $s<k$
  \begin{align*}
\tau_{0,n+1}^{[s+k]}&= (-1)^{s k}   {IW}_{s,k}^*(\widetilde{q}_s,\dots,\widetilde{q}_1;
\Delta(\widetilde{r}_1),\dots,\Delta(\widetilde{r}_k))\cdot\left( \prod_{m=0}^{k-s-1} \frac{\tau_{0,n-m}}{\tau_{1,n-m}} \right)\cdot \tau_{0,n+1},\\
\tau_{1,n+1}^{[s+k]}&= (-1)^{(s+1) k}  {IW}^{*}_{s+1,k}(1,\widetilde{q}_s,\dots,\widetilde{q}_1;
\Delta(\widetilde{r}_1),\dots,\Delta(\widetilde{r}_k))\cdot\left( \prod_{m=0}^{k-s-2} \frac{\tau_{0,n-m}}{\tau_{1,n-m}} \right)\cdot \tau_{0,n+1}.
  \end{align*}
\end{itemize}
The tau pair $(\tau_{0,n}^{[s+k]}, \tau_{1,n}^{[s+k]})$ is a new mToda tau pair. Furthermore, it satisfies the mToda bilinear equation \eqref{HirotamToda}.
\end{proposition}
\begin{proof}
Here we only prove the case for $s\geq k$, as the proof for $s<k$ follows similarly. First, we know from Theorem \ref{Todaeigadjo} that $\tau_{0,n}$ and $\tau_{1,n}$ are two Toda tau functions. If we denote $\tau_n^{\rm Toda}:=\tau_{0,n}$, $\tau_{1,n}:=q_{n}\tau_n^{\rm Toda}$, ${q_{i,n}}:=
\widetilde{q}_{i,n}{q_{n}}$ and
$r_{i,n}:=\frac{\Delta(\widetilde{r}_{i,n})}{q_{n}},$
  it follows from Lemma \ref{muratras}     that $q_n$ and ${q_{i,n}}$ are the Toda eigenfunctions corresponding to $\tau_n^{\rm Toda}$, $r_{i,n}$  are the Toda adjoint eigenfunctions also corresponding to $\tau_n^{\rm Toda}$. Thus if $\tau_{0,n}^{[s+k]}$ is given by
 \begin{align*}
    \tau_{0,n}^{[s+k]} = (-1)^{ks}  {IW}_{k,s}(\Delta(\widetilde{r}_k),\dots,\Delta(\widetilde{r}_1); \widetilde{q}_1,\dots,\widetilde{q}_s)\cdot\left( \prod_{m=0}^{s-k-1} \frac{\tau_{1,n+m}}{\tau_{0,n+m}} \right)\cdot \tau_{0,n}.
    \end{align*}
Therefore, for the first $k$ rows of $IW_{k,s}$ in above equation, we have:
\begin{align}
    \Omega(\widetilde{q}_i, \Delta(\widetilde{r}_j))=\Omega\left( \frac{ {q}_i}{q_n} ,  {q_n}  {r}_j  \right) =\Omega(q_i, r_j),   \label{eq:integral_transform}
\end{align}
where $i=1,2,\cdots,s$ and  $j=1,2,\cdots,k$.
For the remaining $s-k$ rows of $IW_{k,s}$, the entries are $\Delta^m(\frac{  {q}_i}{q_n})$. Using the discrete Leibniz rule $\Delta(A_n B_n) = A_{n+1}\Delta (B_n) +\Delta( A_n)B_n$ and performing consecutive elementary row operations, we have:
\begin{align}
    &\begin{vmatrix}
 		 \frac{ {q}_1}{q_n}& \frac{ {q}_2}{q_n}&\cdots& \frac{ {q}_s}{q_n}\\
 		\Delta( \frac{ {q}_1}{q_n})  &\Delta( \frac{ {q}_2}{q_n})&\cdots&\Delta( \frac{ {q}_s}{q_n})\\
 		\vdots & \vdots& \ddots & \vdots\\
 		\Delta^{s-k-1}( \frac{ {q}_1}{q_n})&\Delta^{s-k-1}( \frac{ {q}_2}{q_n}) &\cdots& \Delta^{s-k-1}( \frac{ {q}_s}{q_n})
 	\end{vmatrix}\nonumber \\
=&\left(\frac{1}{ \prod_{m=0}^{s-k-1} q_{n+m}} \right)\cdot \begin{vmatrix} {q}_1& {q}_2&\cdots& {q}_s\\
 		\Delta( {q}_1)  &\Delta( {q}_2)&\cdots&\Delta( {q}_s)\\
 		\vdots & \vdots& \ddots & \vdots\\
 		\Delta^{s-k-1}( {q}_1)&\Delta^{s-k-1}( {q}_2) &\cdots& \Delta^{s-k-1}( {q}_s)
 	\end{vmatrix}. \label{eq:casoratian_gauge}
\end{align}
Further by \eqref{eq:integral_transform}, \eqref{eq:casoratian_gauge} and $q_n = \tau_{1,n}/\tau_n^{\rm Toda}$, we obtain
\begin{align*}
({\tau_n^{\rm Toda}})^{[s+k]}=\tau_{0,n}^{[s+k]}=(-1)^{sk}{IW}_{k,s}(r_k,\dots,r_1;q_1,\dots,q_s)\tau_n^{\rm Toda}.
\end{align*}
Similarly, we have
\begin{align*}
\tau_{1,n}^{[s+k]}=(-1)^{sk}{IW}_{k,s+1}(r_k,\dots,r_1;q_1,\dots,q_s,q)\tau_n^{\rm Toda}.
\end{align*}
Therefore,
\begin{align*}
q_{n}^{[s+k]}=\left(\frac{\tau_{1,n}}{\tau_n^{\rm Toda}}\right)^{[s+k]}= \frac{{IW}_{k,s+1}(r_{k},\dots,r_{1};q_{1},\dots,q_{s},q)}
{IW_{k,s}(r_{k},\dots,r_{1};q_{1},\dots,q_{s})},
\end{align*}
Thus, we know by Lemma \ref{TodatauDX} that $q_{n}^{[s+k]}$ is  the  Toda eigenfunction corresponding to the Toda tau function $(\tau_n^{\rm Toda})^{[s+k]}$. Finally, by Theorem \ref{Todaeigadjo}, it can be found that the   tau pair $(\tau_{0,n}^{[s+k]}, \tau_{1,n}^{[s+k]})$ is a new mToda tau pair, satisfying the mToda bilinear equation \eqref{HirotamToda}.
\end{proof}

\subsection{The generalized addition formulas  }
Based on the multi-step transformations discussed above, this subsection presents the proof of Theorem   \ref{mTodahigerfayide} and the derivation of the mToda generalized  Fay identities.

First,  the action of $X_a(\mu)$ and $X^*_a(\lambda)$ on mToda tau pair $(\tau_{0,n},\tau_{1,n})$ given by
\begin{align*}
&X_a(\mu)(\tau_{0,n})
=(-1)^{\delta_{a,2}}\Psi_a(\mu^{\eta_a})
\tau_{1,n},\\
&X_a(\mu)(\tau_{1,n})
=\frac{\tau_{1,n}\tau_{1,n+1}
\Delta\Big(\Psi_a(\mu^{\eta_a})\Big)}
{\mu^{{\delta_{a,1}}}\tau_{0,n+1}},\\
&X_a^*(\lambda)(\tau_{0,n})
 =\frac{\tau_{0,n}\tau_{0,n-1}
\Delta\big(\Psi_a^*(n-1,\lambda^{\eta_a})\big)}
{\lambda^{{\delta_{a,1}}}\tau_{1,n-1}},\\
&X_a^*(\lambda)(\tau_{1,n})
 =(-1)^{\delta_{a,1}}\Psi_a^*(\lambda^{\eta_a})
 \tau_{0,n},\quad a=1,2.
\end{align*}
Then by the   ${\bf{ Case\ {  b }}}$  $(b=1, 2)$ of Theorem \ref{mTodaeigentran}, we can get
\begin{itemize}
\item {\bf {Case {\rm 1 }}}\begin{align*}
(\tau^{[1]}_{0,n},\tau^{[1]}_{1,n})
=\Big((-1)^{\delta_{a,2}}X_a(\mu)(\tau_{0,n}),
-\mu^{{\delta_{a,1}}} X_a(\mu)(\tau_{1,n})\Big),\quad a=1,2,
\end{align*}
\item {\bf{Case {\rm 2}}}\begin{align*}
(\tau^{[1]}_{0,n},\tau^{[1]}_{1,n})
=\Big(\lambda^{{\delta_{a,1}}} X_a^*(\lambda)(\tau_{0,n}),
(-1)^{{\delta_{a,1}}}X_a^*(\lambda)(\tau_{1,n})\Big),\quad a=1,2.
\end{align*}
\end{itemize}
Therefore,  for the mToda tau pair $(\tau_{0,n},\tau_{1,n})$, we have
\begin{align}
    \tau_{0,n}^{[s+k]} &= (-1)^{s-l}
    \left( \prod_{\gamma=i}^1 \lambda_\gamma X_1^*(\lambda_\gamma) \prod_{\delta=k-i}^{1} X_2^*(\kappa_\delta) \prod_{\beta=s-l}^{1} X_2(\nu_\beta) \prod_{\alpha=l}^1 X_1(\mu_{\alpha}) \right) (\tau_{0,n}), \label{eq:tau0_vertex} \\
    \tau_{1,n}^{[s+k]} &= (-1)^{s+i}
    \left( \prod_{\gamma=i}^1 X_1^*(\lambda_\gamma) \prod_{\delta=k-i}^{1} X_2^*(\kappa_\delta) \prod_{\beta=s-l}^{1} X_2(\nu_\beta) \prod_{\alpha=l}^1 \mu_{\alpha} X_1(\mu_{\alpha}) \right) (\tau_{1,n}). \label{eq:tau1_vertex}
\end{align}
On the other hand, according to the definitions of the mToda wave functions and adjoint wave functions with respect to the tau pair $(\tau_{0,n} ,\tau_{1,n} )$ in \eqref{mToadwave1} and \eqref{mToadwave2}, let us set the following functions in Proposition \ref{mTodatautra}:
\begin{align*}
    &
    \widetilde{q}_{\alpha,n} = \Psi_1(\mu_\alpha),  \quad
    \widetilde{q}_{l+\beta,n} = \Psi_2( \nu^{-1}_\beta),  \quad 1\leq\alpha \leq l, \quad  1\leq \beta \leq s-l,\\
    &
    \widetilde{r}_{\delta,n} = \Psi_2^*(\kappa^{-1}_\delta ),  \quad
    \widetilde{r}_{k-i+\gamma,n} = \Psi_1^*( \lambda_\gamma),\quad 1\leq\gamma \leq i, \quad 1\leq  \delta \leq k-i,
\end{align*}
where $0\leq i\leq k$, $0\leq l\leq s$.
By   Proposition \ref{mTodatautra} and Theorem \ref{mTodaeigentran},  we can get:
\begin{itemize}
  \item $s\geq k$
\begin{align}
\tau_{\vartheta,n}^{[s+k]}= (-1)^{si+\varepsilon} \frac{\prod_{\gamma=1}^i \lambda_\gamma^{s-k+i-\gamma+1}}{\prod_{\alpha=1}^l \mu_\alpha^{\alpha-1}}
    \cdot \det \!\big( \mathbf{M}^{(\vartheta)} \big) \cdot
    \left( \prod_{m=0}^{s-k-1+\vartheta} \frac{\tau_{1,n+m}}{\tau_{0,n+m}} \right) \cdot \tau_{0,n}. \label{eq:unified_compact1}
\end{align}
  \item  $s < k$
\begin{align}
    \tau_{\vartheta,n+1}^{[s+k]}= {(-1)^{si+\varepsilon+s+(k+1)\vartheta}} \frac{\prod_{\gamma=1}^i \lambda_\gamma^{s-k+i-\gamma+1}}{\prod_{\alpha=1}^l \mu_\alpha^{\alpha-1}}
    \cdot \det \!\big( \widetilde{\mathbf{M}}^{(\vartheta)} \big) \cdot
    \left( \prod_{m=0}^{k-s-1-\vartheta} \frac{\tau_{0,n-m}}{\tau_{1,n-m}} \right) \cdot\tau_{0,n+1}.\label{eq:unified_compact_slk2}
\end{align}
\end{itemize}
Further we  discuss the generalized addition formulas (Fay identities) for the  mToda hierarchy by  the   equivalence between the multi-step Darboux transformations  and the   algebraic vertex operator actions on the mToda tau pair  $(\tau_{0,n},\tau_{1,n})$. By \eqref{eq:tau0_vertex}--\eqref{eq:unified_compact_slk2},  we can get:
\begin{itemize}
  \item $s\geq k$
\begin{align*}
    &\left(\prod_{\gamma=i}^1 X_1^*(\lambda_\gamma) \prod_{\delta=k-i}^{1} X_2^*(\kappa_\delta) \prod_{\beta=s-l}^{1}X_2(\nu_\beta) \prod_{\alpha=l}^1 X_1(\mu_{\alpha})\right)(\tau_{\vartheta,n}) \nonumber \\
    &= (-1)^{s(i-1)+\varepsilon+(1-\vartheta)l + \vartheta i} \frac{\prod_{\gamma=1}^i \lambda_\gamma^{s-k+i-\gamma+\vartheta}}{\prod_{\alpha=1}^l \mu_\alpha^{\alpha-1+\vartheta}}
    \cdot \det \!\big( \mathbf{M}^{(\vartheta)} \big) \cdot
    \left( \prod_{m=0}^{s-k-1+\vartheta} \frac{\tau_{1,n+m}}{\tau_{0,n+m}} \right) \cdot \tau_{0,n},
\end{align*}
  \item   $s < k$
\begin{align*}
    &\left(\prod_{\gamma=i}^1 X_1^*(\lambda_\gamma) \prod_{\delta=k-i}^{1} X_2^*(\kappa_\delta) \prod_{\beta=s-l}^{1}X_2(\nu_\beta) \prod_{\alpha=l}^1 X_1(\mu_{\alpha})\right)(\tau_{\vartheta,n}) \nonumber \\
    &= (-1)^{si+\varepsilon+(1-\vartheta)l -(k+1-i) \vartheta} \frac{\prod_{\gamma=1}^i \lambda_\gamma^{s-k+i-\gamma+\vartheta}}{\prod_{\alpha=1}^l \mu_\alpha^{\alpha-1+\vartheta}}
    \cdot \det \!\big( \widetilde{\mathbf{M}}^{(\vartheta)} \big)_{n-1} \cdot
    \left( \prod_{m=0}^{k-s-1-\vartheta} \frac{\tau_{0,n-m-1}}{\tau_{1,n-m-1}} \right) \cdot\tau_{0,n}.
\end{align*}
\end{itemize}
{\bf This completes the proof of Theorem \ref{mTodahigerfayide}.}

\begin{proposition}\label{mTodaFa}
The mToda tau pair  $(\tau_{0,n},\tau_{1,n})$   satisfies the following relations:
\begin{itemize}
  \item $s\geq k$:
\begin{align}
    \frac{\tau_{\vartheta,n+\epsilon} (\mathbf{t}+ \delta \mathbf{t})}{\tau_{0,n}}
    =& \frac{ (-1)^{ i(s-1)+\varepsilon + \vartheta(i-l)} } {\mathcal{A}_{\mu,\nu,\kappa,\lambda} }
    \big( \prod_{\gamma=1}^i \lambda_\gamma\big)^{\epsilon- 1 + \vartheta}
    \big( \prod_{\beta=1}^{s-l} \nu_\beta \big)^{\epsilon- 1}
    \big( \prod_{\alpha=1}^l \mu_\alpha \big)^{-\epsilon-\vartheta}
    \big( \prod_{\delta=1}^{k-i} \kappa_\delta \big)^{-\epsilon} \nonumber \\[6pt]
    &\cdot \det \!\big( M^{(\vartheta)} \big), \label{tau_unified_forward}
\end{align}
  \item $s<k$:
\begin{align}
    \frac{\tau_{\vartheta,n+\epsilon} (\mathbf{t}+ \delta \mathbf{t})}{\tau_{0,n}}
    =& \frac{ (-1)^{s(i-1)+\varepsilon+i - \vartheta(k-i+l+1)} } {\mathcal{A}_{\mu,\nu,\kappa,\lambda} }
    \big( \prod_{\gamma=1}^i \lambda_\gamma \big)^{\epsilon- 1 + \vartheta}
    \big( \prod_{\beta=1}^{s-l} \nu_\beta \big)^{\epsilon- 1}
    \big( \prod_{\alpha=1}^l \mu_\alpha \big)^{-\epsilon-\vartheta}
    \big( \prod_{\delta=1}^{k-i} \kappa_\delta \big)^{-\epsilon} \nonumber \\[6pt]
    &\cdot \det \!\big(\widetilde{M}^{(\vartheta)} \big)_{n-1}, \label{tau_unified_dual}
\end{align}
\end{itemize}
 where     $\epsilon = (s-l) - (k-i)$, $\delta \mathbf{t} = \sum_{\gamma=1}^i [\lambda_\gamma^{-1}]_1 - \sum_{\alpha=1}^l [\mu_\alpha^{-1}]_1 + \sum_{\delta=1}^{k-i} [\kappa_\delta^{-1}]_2 - \sum_{\beta=1}^{s-l} [\nu_\beta^{-1}]_2$, and $\mathcal{A}_{\mu,\nu,\kappa,\lambda}$ is the  Cauchy-Vandermonde determinant:
\begin{equation*}
\mathcal{A}_{\mu,\nu,\kappa,\lambda} = \frac{\prod_{1 \le \alpha' < \alpha \le l}(\mu_\alpha - \mu_{\alpha'}) \prod_{1 \le \gamma' < \gamma \le i}(\lambda_\gamma - \lambda_{\gamma'}){\prod_{1 \le \beta' < \beta \le s-l}(\nu_\beta  - \nu_{\beta'} ) \prod_{1 \le \delta '< \delta \le k-i}(\kappa_\delta  - \kappa_{\delta'} )}}{\prod_{\alpha=1}^l \prod_{\gamma=1}^i ( \lambda_\gamma-\mu_\alpha) \prod_{\beta=1}^{s-l} \prod_{\delta=1}^{k-i} (  \kappa_\delta-\nu_\beta)}.
\end{equation*}
Here \begin{align*}
    {M}^{(0)} := \begin{pmatrix}   {A} \\  { B} \end{pmatrix}_{s\times s},\quad
    {M}^{(1)} := \begin{pmatrix}  {A} &{D}  \\  {C} & {E}\end{pmatrix}_{(s+1)\times(s+1)},\quad
    \widetilde{{M}}^{(0)} := \begin{pmatrix} \widetilde{A} \\[4pt] \widetilde{ B}\end{pmatrix}_{k\times k},\quad
    \widetilde{{M}}^{(1)} := \begin{pmatrix} \widetilde{D}  \\[4pt] \widetilde{A}  \\[4pt] \widetilde{C}  \end{pmatrix}_{k\times k},
\end{align*} the matrix  sub-blocks are given by:
\begin{small}\begin{align*}
    &A:= \begin{pmatrix}  \Big(\frac{\lambda_{i-\gamma+1}}{\lambda_{i-\gamma+1}- \mu_\alpha} \frac{\tau_{0,n}(\mathbf{t} +[\lambda_{i-\gamma+1}^{-1}]_1-[\mu_\alpha^{-1}]_1)}{\tau_{0,n}(\mathbf{t})}\Big)_{\substack{1 \le \gamma \le i \\ 1 \le \alpha \le l}}&  \Big(\frac{\tau_{0,n+1}(\mathbf{t} +[\lambda_{i-\gamma+1}^{-1}]_1-[\nu_\beta^{-1}]_2)}{\tau_{0,n}(\mathbf{t})} \Big)_{\substack{1 \le \gamma \le i \\ 1 \le \beta \le s-l}} \\[8pt]
    \Big(  \frac{\tau_{0,n-1}(\mathbf{t} +[ \kappa_{k-i-\delta+1}^{-1}]_2-[\mu_\alpha^{-1}]_1)}{\mu_\alpha \kappa_{k-i-\delta+1} \tau_{0,n}(\mathbf{t})}\Big)_{\substack{1 \le \delta \le k-i \\ 1 \le \alpha \le l}} & \Big(\frac{\nu_\beta}{ \kappa_{k-i-\delta+1}-\nu_\beta } \frac{\tau_{0,n}(\mathbf{t} +[ \kappa_{k-i-\delta+1}^{-1}]_2-[\nu_\beta^{-1}]_2)}{\tau_{0,n}(\mathbf{t})}\Big)_{\substack{1 \le \delta \le k-i \\ 1 \le \beta \le s-l}} \end{pmatrix}_{k\times s}, \\[8pt]
    &B:= \begin{pmatrix}  \Big(  \mu_\alpha  ^m  \frac{\tau_{0,n+m}(\mathbf{t}-[\mu_\alpha^{-1}]_1)}{\tau_{0,n+m}(\mathbf{t})} \Big) _{\substack{0 \le m \le s-k-1 \\ 1 \le \alpha \le l}} &\Big(  \nu_\beta^{-m}  \frac{\tau_{0,n+1+m}(\mathbf{t}-[\nu_\beta^{-1}]_2)}{\tau_{0,n+m}(\mathbf{t})} \Big) _{\substack{0 \le m \le s-k-1 \\ 1 \le \beta \le s-l}} \end{pmatrix}_{(s-k)\times s}, \\[8pt]
    &C:= \begin{pmatrix} \Big(  \mu_\alpha^m   \frac{\tau_{0,n+m}(\mathbf{t}-[\mu_\alpha^{-1}]_1)}{\tau_{0,n+m}(\mathbf{t})} \Big) _{\substack{0 \le m \le s-k \\ 1 \le \alpha \le l}}&
    \Big(  \nu_\beta^{-m}   \frac{\tau_{0,n+1+m}(\mathbf{t}-[\nu_\beta^{-1}]_2)}{\tau_{0,n+m}(\mathbf{t})} \Big) _{\substack{0 \le m \le s-k \\ 1 \le \beta \le s-l}}\end{pmatrix}_{(s-k+1)\times s},\\[8pt]
    & D:= \begin{pmatrix} \Big(\frac{\tau_{1,n}(\mathbf{t} +[\lambda_{i-\gamma+1}^{-1}]_1)}{\tau_{0,n}(\mathbf{t})} \Big)_{\substack{1 \le \gamma \le i \\ 1}} \\[8pt] \Big(\kappa_{k-i-\delta+1}^{-1}\frac{\tau_{1,n-1}(\mathbf{t} +[\kappa_{k-i-\delta+1}^{-1}]_2)}{\tau_{0,n}(\mathbf{t})}\Big)_{\substack{1 \le \delta \le k-i \\ 1}} \end{pmatrix}_{k\times 1},  \quad
     {E} :=\begin{pmatrix} \Big(\frac{\tau_{1,n+m}(\mathbf{t})}{\tau_{0,n+m}(\mathbf{t})}\Big) _{\substack{0 \le m \le s-k \\ 1 }}\end{pmatrix}_{(s-k+1)\times 1},\\
   & \widetilde{A}:= \begin{pmatrix}
    \Big( \frac{\nu_{s-l-\beta+1}}{ \kappa_\delta -\nu_{s-l-\beta+1}} \frac{\tau_{0,n+1}(\mathbf{t} +[\kappa_\delta^{-1}]_2    -[\nu_\beta^{-1}]_2)}{\tau_{0,n+1}(\mathbf{t})}\Big)_{\substack{1 \le \beta \le s-l \\ 1 \le \delta \le k-i}} &
    \Big(  \frac{\tau_{0,n+2}(\mathbf{t} +[\lambda_\gamma^{-1}]_1-[\nu_{s-l-\beta+1}^{-1}]_2)}{\tau_{0,n+1}(\mathbf{t})}\Big)_{\substack{1 \le \beta \le s-l \\ 1 \le \gamma \le i}} \\[12pt]
     \Big( \frac{\tau_{0,n}(\mathbf{t} -[\mu_{l-\alpha+1}^{-1}]_1+[\kappa_\delta^{-1}]_2)}{\mu_{l-\alpha+1}\kappa_\delta\tau_{0,n+1}(\mathbf{t})}\Big)_{\substack{1 \le \alpha \le l \\ 1 \le \delta \le k-i}}&
    \Big(\frac{\lambda_\gamma}{\lambda_\gamma- \mu_{l-\alpha+1} } \frac{\tau_{0,n+1}(\mathbf{t} -[\mu_{l-\alpha+1}^{-1}]_1+[\lambda_\gamma^{-1}]_1)}{\tau_{0,n+1}(\mathbf{t})}\Big)_{\substack{1 \le \alpha \le l \\ 1 \le \gamma \le i}}
    \end{pmatrix}_{s\times k}, \\[12pt]
  &  \widetilde{B} := \begin{pmatrix}
   \Big( \kappa_\delta^{-m-1} \frac{\tau_{0,n-m}(\mathbf{t}+[\kappa_\delta^{-1}]_2)}{\tau_{0,n-m+1}(\mathbf{t})} \Big) _{\substack{0 \le m \le k-s-1 \\ 1 \le \delta \le k-i}} & \Big( -\lambda_\gamma^{m+1} \frac{\tau_{0,n-m+1}(\mathbf{t}+[\lambda_\gamma^{-1}]_1)}{\tau_{0,n-m+1}(\mathbf{t})} \Big) _{\substack{0 \le m \le k-s-1 \\ 1 \le \gamma \le i}}
    \end{pmatrix}_{(k-s)\times k}, \\[12pt]
    &\widetilde{C} := \begin{pmatrix}
 \Big( \kappa_\delta^{-m-1} \frac{\tau_{0,n-m}(\mathbf{t}+[\kappa_\delta^{-1}]_2)}{\tau_{0,n-m+1}(\mathbf{t})} \Big) _{\substack{0 \le m \le k-s-2 \\ 1 \le \delta \le k-i}} & \Big(- \lambda_\gamma^{m+1} \frac{\tau_{0,n-m+1}(\mathbf{t}+[\lambda_\gamma^{-1}]_1)}{\tau_{0,n-m+1}(\mathbf{t})} \Big) _{\substack{0 \le m \le k-s-2 \\ 1 \le \gamma \le i}}
    \end{pmatrix}_{(k-s-1)\times k},\\[12pt]
    &\widetilde{D}:= \begin{pmatrix}
 \Big(\frac{\tau_{1,n}(\mathbf{t}+[\kappa_\delta^{-1}]_2)}
    {\kappa_\delta\tau_{0,n+1}(\mathbf{t})}\Big)_{\substack{  1 \\1 \le \delta \le k-i}}&   \Big(  \frac{\tau_{1,n+1}(\mathbf{t}+[\lambda_\gamma^{-1}]_1)}{\tau_{0,n+1}(\mathbf{t})} \Big)_{\substack{1 \\ 1 \le \gamma \le i }}
    \end{pmatrix}_{1\times k}.
\end{align*}
\end{small}
\end{proposition}
\begin{proof}First, the  vertex operator actions on the mToda tau pair  $(\tau_{0,n},\tau_{1,n})$,
\begin{align}
 &\left(\prod_{\gamma=i}^1 X_1^*(\lambda_\gamma) \prod_{\delta=k-i}^{1} X_2^*(\kappa_\delta) \prod_{\beta=s-l}^{1}X_2(\nu_\beta) \prod_{\alpha=l}^1 X_1(\mu_{\alpha})\right)(\tau_{\vartheta,n}) \nonumber\\
=& (-1)^{s-l+i}\mathcal{A}_{\mu,\nu,\kappa,\lambda}\cdot \big( \prod_{\gamma=1}^i \lambda_\gamma^{-n-\gamma+l+1}\big) \big( \prod_{\delta=1}^{k-i} \kappa_\delta^{n+\epsilon} \big)
\cdot \big( \prod_{\alpha=1}^l \mu_\alpha^{n-\alpha+\epsilon+1} \big) \big( \prod_{\beta=1}^{s-l} \nu_\beta^{-n-\epsilon+1} \big) \nonumber \\
&\cdot e^{ \xi (\mathbf{t}^{(1)}, \sum_{\alpha=1}^l \mu_\alpha - \sum_{\gamma=1}^i \lambda_\gamma)} e^{\xi(\mathbf{t}^{(2)}, \sum_{\beta=1}^{s-l} \nu_\beta - \sum_{\delta=1}^{k-i} \kappa_\delta)} \cdot \tau_{\vartheta,n+\epsilon}  \left(\mathbf{t}+\delta \mathbf{t}\right).\label{taumutidt}
\end{align}
On the other hand,  by Lemma \ref{Lem:eq}, we have:
\begin{align*}
&\Omega\Big(\Psi_1(n,\mathbf{t},\mu_\alpha),
\Delta\big(\Psi^*_1(n,\mathbf{t},\lambda_\gamma)\big)\Big)
=\Psi_1(n,\mathbf{t}+[\lambda^{-1}_\gamma]_1,\mu_\alpha)
\Psi^*_1(n,\mathbf{t},\lambda_\gamma),\\
&\Omega\Big(\Psi_2(n,\mathbf{t},\nu^{-1}_\beta),
\Delta\big(\Psi^*_1(n,\mathbf{t},\lambda_\gamma)\big)\Big)
=\Psi_2(n,\mathbf{t}+[\lambda^{-1}_\gamma]_1,\nu^{-1}_\beta)
\Psi^*_1(n,\mathbf{t},\lambda_\gamma),\\
&\Omega\Big(\Psi_1(n,\mathbf{t},\mu_\alpha),
\Delta\big(\Psi^*_2(n,\mathbf{t},\kappa^{-1}_\delta)\big)\Big)
=\Psi_1(n-1,\mathbf{t}+[\kappa^{-1}_\delta]_2,\mu_\alpha)
\Psi^*_2(n,\mathbf{t},\kappa^{-1}_\delta),\\
&\Omega\Big(\Psi_2(n,\mathbf{t},\nu^{-1}_\beta),
\Delta\big(\Psi^*_2(n,\mathbf{t},\kappa^{-1}_\delta)\big)\Big)
=\Psi_2(n-1,\mathbf{t}+[\kappa^{-1}_\delta]_2,\nu^{-1}_\beta)
\Psi^*_2(n,\mathbf{t},\kappa^{-1}_\delta).
\end{align*}
Thus, we can get
\begin{align}
    \det \!\big( \mathbf{M}^{(\vartheta)} \big) &= \Gamma_n(\mu,\nu,\lambda,\kappa) \cdot \left( \prod_{m=0}^{s-k-1+\vartheta} \frac{\tau_{0,n+m}}{\tau_{1,n+m}} \right) \cdot \det \!\big({M}^{(\vartheta)} \big), \label{eq:det_M_forward} \\[8pt]
    \det \!\big( \widetilde{\mathbf{M}}^{(\vartheta)} \big) &= \Gamma_{n+1}(\mu,\nu,\lambda,\kappa) \cdot \left( \prod_{m=0}^{k-s-1-\vartheta} \frac{\tau_{1,n-m}}{\tau_{0,n-m}} \right) \cdot \det \!\big(\widetilde{M}^{(\vartheta)} \big), \label{eq:det_M_dual}
\end{align}
where $\Gamma_n(\mu,\nu,\lambda,\kappa) :=  \prod_{\alpha=1}^l \mu_\alpha^n e^{\xi(\mathbf{t}^{(1)}, \mu_\alpha)} \prod_{\beta=1}^{s-l} \nu_\beta^{-n} e^{\xi(\mathbf{t}^{(2)}, \nu_\beta)} \prod_{\gamma=1}^i \lambda_\gamma^{-n} e^{-\xi(\mathbf{t}^{(1)}, \lambda_\gamma)} \prod_{\delta=1}^{k-i} \kappa_\delta^n e^{-\xi(\mathbf{t}^{(2)}, \kappa_\delta)}$.
Substituting  \eqref{taumutidt}--\eqref{eq:det_M_dual} into Theorem \ref{mTodahigerfayide}, we can get \eqref{tau_unified_forward} and \eqref{tau_unified_dual}.

\end{proof}

\begin{remark}
 According to Theorem \ref{Todaeigadjo},    the mToda generalized Fay identities can be viewed as a higher-dimensional generalization of the  Toda hierarchy.  For example, by imposing   parameter constraints  (e.g., setting $s=l=i=k$ and $s=l=k$, $i=0$ in \eqref{tau_unified_forward}, respectively), the  mToda generalized Fay identities  degenerate into the Toda   Fay identities, that is
\begin{align*}
&{\rm det}\left(\frac{\tau_{0,n}(\mathbf{t}-[u_i^{-1}]_1+[\lambda_j^{-1}]_1)}
{(\lambda_j-\mu_i)\tau_{0,n}(\mathbf{t})}\right)_{1\leq i,j\leq k}
=(-1)^{\frac{k(k-1)}{2}}\frac{\Delta(\lambda)\Delta(\mu)
}{\Pi_{i,j=1}^k(\lambda_i-\mu_j)}
\frac{\tau_{0,n}(\mathbf{t}+\sum_{j=1}^k[\lambda_j^{-1}]_1
-\sum_{i=1}^k[\mu_i^{-1}]_1)}
{\tau_{0,n}(\mathbf{t})},\\
&{\rm det}\left(\frac{\tau_{0,n-1}(\mathbf{t}-[u_i^{-1}]_1+[\lambda_j^{-1}]_2)}
{\tau_{0,n}(\mathbf{t})}\right)_{1\leq i,j\leq k}
={\Delta(\lambda^{-1})\Delta(\mu^{-1})
}
\frac{\tau_{0,n-k}(\mathbf{t}-\sum_{i=1}^k[\mu_i^{-1}]_1+\sum_{j=1}^k[\lambda_j^{-1}]_2
)}
{\tau_{0,n}(\mathbf{t})}.
\end{align*}
This proves that the core results in  such as Adler and van Moerbeke \cite{Adler1999} are actually specific, lower-dimensional projections of Theorem \ref{mTodahigerfayide} and this proposition.
 \end{remark}

\begin{example}
Here we present several examples in \eqref{tau_unified_forward} and \eqref{tau_unified_dual}.

    \paragraph{ {\bf {Examples for $\tau_{0,n}$:}}}
 \begin{itemize}
    \item $s=2, k=0, l=2, i=0$,
    \begin{align*}
        &(\mu_1 - \mu_2) \tau_{0,n+1}(\mathbf{t}) \tau_{0,n}(\mathbf{t}-[\mu_1^{-1}]_1-[\mu_2^{-1}]_1) \\
        &= \mu_1 \tau_{0,n}(\mathbf{t}-[\mu_2^{-1}]_1)\tau_{0,n+1}(\mathbf{t}-[\mu_1^{-1}]_1) - \mu_2 \tau_{0,n}(\mathbf{t}-[\mu_1^{-1}]_1)\tau_{0,n+1}(\mathbf{t}-[\mu_2^{-1}]_1).
    \end{align*}
 \item  $s=2, k=0, l=1, i=0$,
    \begin{align*}
        &\tau_{0,n}(\mathbf{t}-[\mu_1^{-1}]_1) \tau_{0,n+2}(\mathbf{t}-[\nu_1^{-1}]_2) \\
        &= \mu_1 \nu_1 \Big( \tau_{0,n+1}(\mathbf{t}-[\nu_1^{-1}]_2) \tau_{0,n+1}(\mathbf{t}-[\mu_1^{-1}]_1) - \tau_{0,n+1}(\mathbf{t}) \tau_{0,n+1}(\mathbf{t} - [\mu_1^{-1}]_1 - [\nu_1^{-1}]_2) \Big).
    \end{align*}
    \item   $s=2, k=0, l=0, i=0$,
    \begin{align*}
        &(\nu_1 - \nu_2) \tau_{0,n}(\mathbf{t}) \tau_{0,n+1}(\mathbf{t}) \tau_{0,n+2}(\mathbf{t} - [\nu_1^{-1}]_2 - [\nu_2^{-1}]_2) \\
        &= \nu_1 \tau_{0,n+1}(\mathbf{t}-[\nu_1^{-1}]_2)\tau_{0,n+2}(\mathbf{t}-[\nu_2^{-1}]_2) - \nu_2 \tau_{0,n+1}(\mathbf{t}-[\nu_2^{-1}]_2)\tau_{0,n+2}(\mathbf{t}-[\nu_1^{-1}]_2).
    \end{align*}

   \item  $s=2, k=1$, $l=2, i=1$, $\vartheta=0$,
\begin{align*}
    0 =& (\mu_1 - \mu_2)\tau_{0,n}(\mathbf{t} + [\mu_1^{-1}]_1 + [\mu_2^{-1}]_1)\tau_{0,n}(\mathbf{t} + [\lambda_1^{-1}]_1) \\
    &+ (\mu_2 - \lambda_1)\tau_{0,n}(\mathbf{t} + [\lambda_1^{-1}]_1 + [\mu_2^{-1}]_1)\tau_{0,n}(\mathbf{t} + [\mu_1^{-1}]_1) \\
    &+ (\lambda_1 - \mu_1)\tau_{0,n}(\mathbf{t} + [\lambda_1^{-1}]_1 + [\mu_1^{-1}]_1)\tau_{0,n}(\mathbf{t} + [\mu_2^{-1}]_1).
\end{align*}
This equation is the discrete KP (Hirota-Miwa) equation \cite{Savchenko20261,Miwa1982}.
   \item $s=2, k=1$, $l=1, i=0$,
\begin{align*}
    &\mu_1 \kappa_1 \nu_1 \tau_{0,n}(\mathbf{t}) \tau_{0,n}(\mathbf{t} - [\mu_1^{-1}]_1 + [\kappa_1^{-1}]_2 - [\nu_1^{-1}]_2) \nonumber \\
    &= \mu_1 \kappa_1 \nu_1 \tau_{0,n}(\mathbf{t}+[\kappa_1^{-1}]_2-[\nu_1^{-1}]_2)\tau_{0,n}(\mathbf{t}-[\mu_1^{-1}]_1) \nonumber \\
    &\quad - (\kappa_1-\nu_1) \tau_{0,n-1}(\mathbf{t}+[\kappa_1^{-1}]_2-[\mu_1^{-1}]_1)\tau_{0,n+1}(\mathbf{t}-[\nu_1^{-1}]_2).
\end{align*}

   \item $s=2, k=2$, $l=2, i=2$,
 \begin{align*}
        &(\mu_2-\mu_1)(\lambda_2-\lambda_1) \tau_{0,n}(\mathbf{t}) \tau_{0,n}(\mathbf{t}+[\lambda_1^{-1}]_1+[\lambda_2^{-1}]_1-[\mu_1^{-1}]_1-[\mu_2^{-1}]_1) \\
        &= (\lambda_1-\mu_1)(\lambda_2-\mu_2) \tau_{0,n}(\mathbf{t}+[\lambda_2^{-1}]_1-[\mu_1^{-1}]_1)\tau_{0,n}(\mathbf{t}+[\lambda_1^{-1}]_1-[\mu_2^{-1}]_1) \\
        &\quad - (\lambda_1-\mu_2)(\lambda_2-\mu_1) \tau_{0,n}(\mathbf{t}+[\lambda_2^{-1}]_1-[\mu_2^{-1}]_1)\tau_{0,n}(\mathbf{t}+[\lambda_1^{-1}]_1-[\mu_1^{-1}]_1).
    \end{align*}
   \item $s=1, k=2$, $l=1, i=2$,
   \begin{align*}
    &(\lambda_2-\lambda_1) \tau_{0,n}(\mathbf{t}) \tau_{0,n}(\mathbf{t}+[\lambda_1^{-1}]_1+[\lambda_2^{-1}]_1-[\mu_1^{-1}]_1) \nonumber \\
    &= (\lambda_2-\mu_1) \tau_{0,n}(\mathbf{t}-[\mu_1^{-1}]_1+[\lambda_1^{-1}]_1)\tau_{0,n}(\mathbf{t}+[\lambda_2^{-1}]_1) \nonumber \\
    &\quad - (\lambda_1-\mu_1) \tau_{0,n}(\mathbf{t}-[\mu_1^{-1}]_1+[\lambda_2^{-1}]_1)\tau_{0,n}(\mathbf{t}+[\lambda_1^{-1}]_1).
\end{align*}
 \end{itemize}

 \paragraph{ {\bf {Examples for $\tau_{1,n}$:}}}
  \begin{itemize}
     \item  $s=1, k=1, l=1, i=1$
    \begin{align*}
        &\mu_1 \tau_{0,n}(\mathbf{t}) \tau_{1,n}(\mathbf{t} + [\lambda_1^{-1}]_1 - [\mu_1^{-1}]_1) \\
        &= \lambda_1 \tau_{0,n}(\mathbf{t} + [\lambda_1^{-1}]_1 - [\mu_1^{-1}]_1) \tau_{1,n}(\mathbf{t}) - (\lambda_1 - \mu_1) \tau_{0,n}(\mathbf{t} - [\mu_1^{-1}]_1) \tau_{1,n}(\mathbf{t} + [\lambda_1^{-1}]_1).
    \end{align*}
   \item $s=1, k=2$, $l=1, i=2$
   \begin{align*}
    &\mu_1 (\lambda_2-\lambda_1) \tau_{0,n}(\mathbf{t}) \tau_{1,n}(\mathbf{t}+[\lambda_1^{-1}]_1+[\lambda_2^{-1}]_1-[\mu_1^{-1}]_1) \\
    &= \lambda_1(\lambda_2-\mu_1) \tau_{0,n}(\mathbf{t}-[\mu_1^{-1}]_1+[\lambda_1^{-1}]_1)\tau_{1,n}(\mathbf{t}+[\lambda_2^{-1}]_1) \\
    &\quad - \lambda_2(\lambda_1-\mu_1) \tau_{0,n}(\mathbf{t}-[\mu_1^{-1}]_1+[\lambda_2^{-1}]_1)\tau_{1,n}(\mathbf{t}+[\lambda_1^{-1}]_1).
\end{align*}
   \item $s=1, k=2$, $l=0, i=1$,
\begin{align*}
    &\nu_1 \kappa_1 \tau_{0,n}(\mathbf{t}) \tau_{1,n}(\mathbf{t}+[\lambda_1^{-1}]_1+[\kappa_1^{-1}]_2-[\nu_1^{-1}]_2) \\
    &= \nu_1 \kappa_1 \tau_{0,n}(\mathbf{t}+[\kappa_1^{-1}]_2-[\nu_1^{-1}]_2)\tau_{1,n}(\mathbf{t}+[\lambda_1^{-1}]_1) \\
    &\quad - (\kappa_1-\nu_1) \tau_{0,n+1}(\mathbf{t}+[\lambda_1^{-1}]_1-[\nu_1^{-1}]_2)\tau_{1,n-1}(\mathbf{t}+[\kappa_1^{-1}]_2).
\end{align*}

\end{itemize}
All the above cases correspond to the Fay identities of the mToda hierarchy. In particular, the derivations for $\tau_{0,n}$ also yield the  Fay identities of the Toda hierarchy \cite{Cheng2013,Teo2006}.

 \end{example}

\section{Conclusion and discussion}\label{section5}
In this paper, we have  investigated the  tau functions transformations and symmetries of the mToda hierarchy. The main results are summarized as follows:

\begin{itemize}
    \item \textbf{Relationship between Toda and mToda  tau  functions:} We have shown that the mToda tau pair $(\tau_{0,n},\tau_{1,n})$ can be regarded as two Toda tau functions satisfying the Toda bilinear equation, with their ratios corresponding to the Toda eigenfunctions and adjoint eigenfunctions        in Theorem \ref{Todaeigadjo}. Further, by the mToda eigenfunctions and adjoint eigenfunctions, we constructed the    transformations for the mToda  tau  pair, Lax operators, and wave functions in Theorem \ref{mTodaeigentran}. Moreover, we demonstrated how to generate new mToda  tau  pair  using the squared eigenfunction potential.

    \item \textbf{Symmetries and the ASvM formula:} We introduced the mToda squared eigenfunction symmetries  and proved that  the mToda additional flows $\pa^{*}_{\widetilde{q}_{n},\widetilde{r}_{n}}$ commute with  the time flows  in  Proposition \ref{pro1}. Via the actions of vertex operators, we   derived the ASvM formula, which connects  the additional flows on the wave functions with the Sato--B\"{a}cklund transformations of the tau functions in   Theorem \ref{ASvM}.

    \item \textbf{Multi-step  transformations and Fay identities:} We derived the mToda generalized addition  formulas (Fay identities)   by  the   equivalence between the multi--step Darboux transformations  and the   vertex operator actions on the  tau pair  $(\tau_{0,n},\tau_{1,n})$ in Theorem \ref{mTodahigerfayide}.  Further, the mToda generalized addition formulas (Fay identities)  can degenerate into classical Toda Fay identities under  parameter constraints in Proposition \ref{mTodaFa}.
\end{itemize}

In the context of integrable systems and coupled random matrices, the actions of vertex operators on Toda tau functions generate the Christoffel--Darboux kernels and Fredholm determinants \cite{Adler1999}. Furthermore, under the vacuum background reduction, the generalized Wronskian determinants constructed in this paper can  degenerate into Jacobi--Trudi identities. This  reduction provides an explicit algebraic framework to construct rational and polynomial solutions for the mToda hierarchy via Schur functions. Therefore, based on the determinant formulas in Theorem \ref{mTodahigerfayide}, we will further investigate  the   structures of the mToda hierarchy to expand the connections between integrable systems, representation theory, and random matrix theory.\\

\section{Appendix}
In this section, we will show that the squared eigenfunction potential (SEP) $\Omega_n := \Omega(f_n,g_n)$ is well-defined, where
\begin{align}
	\partial_{t_k^{(a)}} f_n &= B_k^{(a)}(f_n), \quad \partial_{t_k^{(a)}} g_n = -B_k^{(a)*}(g_n), \quad a = 1, 2, \label{eq:evolution_fg}
\end{align}
and  $B_k^{(a)}$ satisfy the zero-curvature equations
\begin{align}
	\partial_{t_l^{(b)}} B_k^{(a)} - \partial_{t_k^{(a)}} B_l^{(b)} + [B_k^{(a)}, B_l^{(b)}] &= 0, \quad b=1,2.\label{eq:zero_curvature}
\end{align}

 Before proving this, we introduce the following lemma.

\begin{lemma}\label{omedade}
For any difference operator $Q = \sum_{i=1}^{N} a_{i,n} P^i$ with $P \in \{\Delta, \Delta^*\}$, letting $f:=f_n(\mathbf{t})$ and $g:=g_n(\mathbf{t})$ be arbitrary functions, the following identity holds:
\begin{equation*}
    Q(f)g - f Q^*(g) = P \left( \mathrm{Res}_{P} \left( P^{-1} g Q f P^{-1} \right) \right),
\end{equation*}
where $\operatorname{Res}_P\sum_i a_i P^i=a_{-1}$.
\end{lemma}

\begin{proof}
Here we only prove the case for $P = \Delta$, as the proof for $P = \Delta^*$ follows similarly.
Let $A = g Q f$, which can be expanded as $A = \sum_{j=0}^{M} b_{j,n} \Delta^j$. Thus,
\begin{align*}
    Q(f)g - f Q^*(g) &= (g Q f)(1) - (f Q^* g)(1) = A(1) - A^*(1) = \sum_{j=1}^M \Delta(-\Delta)^{j-1}(b_{j,n-j}),
\end{align*}
where we have used $(\Delta^*)^j(b_n) = (-1)^j \Delta^j(b_{n-j})$.
On the other hand, by   $\Delta^{-1} \cdot b_{j,n} = \sum_{i=1}^{\infty} (-\Delta)^{i-1}(b_{j,n-i}) \Delta^{-i}$, we obtain:
\begin{align*}
    \Delta \left( \mathrm{Res}_{\Delta} \left( \Delta^{-1} g Q f \Delta^{-1} \right) \right)  = \Delta \Big( \mathrm{Res}_{\Delta} \big( \Delta^{-1} \sum_{j=0}^{M} b_{j,n} \Delta^{j-1} \big) \Big)  = \Delta \Big( \sum_{j=1}^{M} (-\Delta)^{j-1} b_{j,n-j} \Big).
\end{align*}
This completes the proof.
\end{proof}

\begin{lemma}\label{lemma:compatibility}
For arbitrary functions $f_n$ and $g_n$ satisfying \eqref{eq:evolution_fg}, we have
\begin{align}
    \partial_{t_k^{(1)}}(f_n g_n) &= \Delta \left( \mathrm{Res}_{\Delta}\big(\Delta^{-1} g_n B_k^{(1)} f_n \Delta^{-1}\big) \right), \label{eq:comp_t1} \\
    \partial_{t_k^{(2)}}(f_n g_n) &= - \Delta \Lambda^{-1} \left( \mathrm{Res}_{\Delta^*}\big(\Delta^{*-1} g_n B_k^{(2)} f_n \Delta^{*-1}\big) \right). \label{eq:comp_t2}
\end{align}
\end{lemma}

\begin{proof}
First, by \eqref{eq:evolution_fg}, we have $\partial_{t_k^{(a)}}(f_n g_n) = B_k^{(a)}(f_n)g_n - f_n B_k^{(a)*}(g_n)$ for $a=1,2$.
Applying Lemma \ref{omedade} with $P=\Delta$ and $Q=B_k^{(1)}$ for $a=1$, and similarly with $P=\Delta^*$ and $Q=B_k^{(2)}$   with  $\Delta^*=-\Delta \Lambda^{-1}$ for $a=2$  yields \eqref{eq:comp_t1} and \eqref{eq:comp_t2}.
Further, by direct computation:
\begin{align*}
    [\partial_{t_k^{(a)}}, \partial_{t_l^{(b)}}](f_n g_n) &= \left( \partial_{t_k^{(a)}}B_l^{(b)} - \partial_{t_l^{(b)}}B_k^{(a)} + [B_l^{(b)}, B_k^{(a)}] \right)(f_n)g_n \\
    &\quad - f_n\left( \partial_{t_k^{(a)}}B_l^{(b)} - \partial_{t_l^{(b)}}B_k^{(a)} + [B_l^{(b)}, B_k^{(a)}] \right)^*(g_n),
\end{align*}
according to \eqref{eq:zero_curvature}, we can get  $[\partial_{t_k^{(a)}}, \partial_{t_l^{(b)}}](f_n g_n)=0$.
\end{proof}
Due to Lemma \ref{lemma:compatibility}, we can define the squared eigenfunction potential $\Omega_n=\Omega(f_n,g_n)$ by:
\begin{align*}
     &\Delta \Omega_n= f_n g_n, \\
    &\partial_{t_k^{(1)}} \Omega_n = \mathrm{Res}_{\Delta}\left(\Delta^{-1} g_n B_k^{(1)} f_n \Delta^{-1}\right), \\
    &\partial_{t_k^{(2)}} \Omega_{n+1} = -\mathrm{Res}_{\Delta^*}\left(\Delta^{*-1} g_n B_k^{(2)} f_n \Delta^{*-1}\right).
\end{align*}\\

\noindent{\bf Conflict of Interest}:
The author has no conflicts to disclose.\\

\noindent{\bf Consent to Participate declaration}: not applicable.\\

\noindent{\bf Consent to Publish declaration}: not applicable.\\

\noindent{\bf Data availability}:
Data sharing is not applicable to this article as no new data were created or analyzed in this study.\\

\noindent{\bf Ethics declaration}: not applicable.\\

\noindent{\bf Funding declaration}:
This work is supported by the National Natural Science Foundation of China
(Grant Nos. 12571271 and 12261072) and   Huaqiao University Research Startup Funds (Nos. 26BS117).

\end{document}